\documentclass[12pt,onecolumn]{IEEEtran}
\newif\iflong
\longtrue

\makeatletter
\usepackage[oldcommands,boxed]{algorithm2e}
\usepackage{amsmath, latexsym, amssymb, amsthm, amsfonts}
\usepackage[mathscr]{eucal}
\usepackage{graphicx}
\usepackage{epstopdf}
\usepackage{epsfig}
\usepackage{float}
\usepackage{setspace}
\usepackage{amsfonts}
\usepackage{psfrag}
\usepackage{comment}
\usepackage{cleveref}

\usepackage{url}
\usepackage{wrapfig}
\usepackage{units}

\usepackage{psfrag}
\usepackage{empheq}

\usepackage{pdflscape}
\usepackage{nicefrac}

\usepackage{wrapfig}
\usepackage{wasysym}

\usepackage{xr}
\usepackage{color}
\usepackage[subnum]{cases}

\usepackage{slashbox}

\newtheorem{lemma}{Lemma}

\newtheorem{definition}{Definition}

\newtheorem{remark}{Remark}

\begin{document}

\def\thelemma{\arabic{section}.\arabic{lemma}}
\def\thetheorem{\arabic{section}.\arabic{theorem}}
\def\thecorollary{\arabic{section}.\arabic{corollary}}
\def\thedefinition{\arabic{section}.\arabic{definition}}
\def\theexample{\arabic{section}.\arabic{example}}
\def\theproposition{\arabic{section}.\arabic{proposition}}
\def\thecondition{\arabic{section}.\arabic{condition}}
\def\theassumption{\arabic{section}.\arabic{assumption}}
\def\theconjecture{\arabic{section}.\arabic{conjecture}}
\def\theproblem{\arabic{section}.\arabic{problem}}
\def\theremark{\arabic{section}.\arabic{remark}}

\newcommand{\manualnames}[1]{
\def\thelemma{#1.\arabic{lemma}}
\def\thetheorem{#1.\arabic{theorem}}
\def\thecorollary{#1.\arabic{corollary}}
\def\thedefinition{#1.\arabic{definition}}
\def\theexample{#1.\arabic{example}}
\def\theproposition{#1.\arabic{proposition}}
\def\theassumption{#1.\arabic{assumption}}
\def\theremark{#1.\arabic{remark}}
}

\newcommand{\beginsec}{
\setcounter{lemma}{0}
\setcounter{theorem}{0}
\setcounter{corollary}{0}
\setcounter{definition}{0}
\setcounter{example}{0}
\setcounter{proposition}{0}
\setcounter{condition}{0}
\setcounter{assumption}{0}
\setcounter{conjecture}{0}
\setcounter{problem}{0}
\setcounter{remark}{0}
}
\newcommand{\la}{\lambda}
\newcommand{\eps}{\varepsilon}
\newcommand{\ph}{\varphi}
\newcommand{\vr}{\varrho}
\newcommand{\al}{\alpha}
\newcommand{\bet}{\beta}
\newcommand{\gam}{\gamma}
\newcommand{\kap}{\kappa}
\newcommand{\s}{\sigma}
\newcommand{\sig}{\sigma}
\newcommand{\om}{\omega}
\newcommand{\Gam}{\mathnormal{\Gamma}}
\newcommand{\off}[1]{}
\newcommand{\Del}{\mathnormal{\Delta}}
\newcommand{\Th}{\mathnormal{\Theta}}
\newcommand{\La}{\mathnormal{\Lambda}}
\newcommand{\X}{\mathnormal{\Xi}}
\newcommand{\PI}{\mathnormal{\Pi}}
\newcommand{\Sig}{\mathnormal{\Sigma}}
\newcommand{\Ups}{\mathnormal{\Upsilon}}
\newcommand{\Ph}{\mathnormal{\Phi}}
\newcommand{\Ps}{\mathnormal{\Psi}}
\newcommand{\Om}{\mathnormal{\Omega}}

\newcommand{\D}{{\mathbb D}}
\newcommand{\M}{{\mathbb M}}
\newcommand{\N}{{\mathbb N}}
\newcommand{\Q}{{\mathbb Q}}
\newcommand{\R}{{\mathbb R}}
\newcommand{\U}{{\mathbb U}}
\newcommand{\Z}{{\mathbb Z}}
\newcommand{\T}{{\mathbb T}}
\newcommand{\A}{{\mathbb A}}
\newcommand{\HH}{{\mathbb H}}

\newcommand{\EE}{{\mathbb E}}
\newcommand{\FF}{{\mathbb F}}
\newcommand{\PP}{{\mathbb P}}
\newcommand{\ONE}{\boldsymbol{1}}

\newcommand{\calA}{{\cal A}}
\newcommand{\calB}{{\cal B}}
\newcommand{\calC}{{\cal C}}
\newcommand{\calD}{{\cal D}}
\newcommand{\calE}{{\cal E}}
\newcommand{\calF}{{\cal F}}
\newcommand{\calG}{{\cal G}}
\newcommand{\calH}{{\cal H}}
\newcommand{\calI}{{\cal I}}
\newcommand{\calJ}{{\cal J}}
\newcommand{\calL}{{\cal L}}
\newcommand{\calM}{{\cal M}}
\newcommand{\calN}{{\cal N}}
\newcommand{\calP}{{\cal P}}
\newcommand{\calR}{{\cal R}}
\newcommand{\calS}{{\cal S}}
\newcommand{\calT}{{\cal T}}
\newcommand{\calU}{{\cal U}}
\newcommand{\calZ}{{\cal Z}}
\newcommand{\calV}{{\cal V}}
\newcommand{\calX}{{\cal X}}
\newcommand{\calY}{{\cal Y}}
\newcommand{\calq}{{\cal q}}

\newcommand{\bI}{{\mathbf I}}
\newcommand{\bD}{{\mathbf D}}
\newcommand{\bJ}{{\mathbf J}}
\newcommand{\bK}{{\mathbf K}}
\newcommand{\bT}{{\mathbf T}}
\newcommand{\bq}{{\mathbf q}}
\newcommand{\bu}{{\mathbf u}}
\newcommand{\bv}{{\mathbf v}}
\newcommand{\bo}{{\mathbf o}}
\newcommand{\bee}{{\mathbf e}}
\newcommand{\bi}{{\mathbf i}}
\newcommand{\bn}{{\mathbf n}}
\newcommand{\bff}{{\mathbf f}}
\newcommand{\bd}{{\mathbf d}}
\newcommand{\bs}{{\mathbf s}}
\newcommand{\bb}{{\mathbf b}}
\newcommand{\ba}{{\mathbf a}}

\newcommand{\scrA}{\mathscr{A}}
\newcommand{\scrM}{\mathscr{M}}
\newcommand{\scrS}{\mathscr{S}}
\newcommand{\scrU}{\mathscr{U}}
\newcommand{\scrI}{\mathscr{I}}
\newcommand{\scrP}{\mathscr{P}}

\newcommand{\mq}{\mathfrak q}

\newcommand{\frA}{\mathfrak{A}}
\newcommand{\frM}{\mathfrak{M}}
\newcommand{\frS}{\mathfrak{S}}

\renewcommand{\proof}{\noindent{\bf Proof.\ }}

\newcommand{\lan}{\langle}
\newcommand{\ran}{\rangle}
\newcommand{\uu}{\underline}
\newcommand{\oo}{\overline}
\newcommand{\supp}{{\rm supp}}
\newcommand{\diag}{{\rm diag}}
\newcommand{\trace}{{\rm trace}}
\newcommand{\w}{\wedge}
\newcommand{\lt}{\left}
\newcommand{\rt}{\right}
\newcommand{\pl}{\partial}
\newcommand{\abs}[1]{\lvert#1\rvert}
\newcommand{\norm}[1]{\lVert#1\rVert}
\newcommand{\mean}[1]{\langle#1\rangle}
\newcommand{\To}{\Rightarrow}
\newcommand{\til}{\widetilde}
\newcommand{\wh}{\widehat}
\newcommand{\dist}{{\rm dist}}
\newcommand{\grad}{\nabla}
\newcommand{\iy}{\infty}
\newcommand{\AddedTh}[1]{\textbf{\textcolor{Black}{#1}}}

\newcommand{\be}{\begin{equation}}
\newcommand{\ee}{\end{equation}}
\newcommand{\mmme}[1]{{\textcolor{green}{MSh: #1}}}

\newcommand{\tab}{\hspace*{0.3in}}
\newcommand{\Tab}{\hspace*{1.0in}}
\newcommand{\no}{\nonumber}
\newcommand{\noi}{\noindent}
\newcommand{\txt}{\textrm}
\newcommand{\ds}{\displaystyle}
\newcommand{\RR}{\mathbb{R}}
\newcommand{\vf}{\varphi}
\newcommand{\del}{\frac{\partial}{\partial t}}
%\definecolor{co}{rgb}{0.8,0,0.8}
%\definecolor{gr}{gray}{0.5}
\newcommand{\gr}{\color{gr}}
\newcommand{\vp}{\varepsilon}
\newcommand{\E}{{\mathbb E}}

\newcommand{\FMO}[1]{\textbf{\textcolor{red}{***FIXME OMER #1 ***}}}

\parskip 0pt
\title{ VM Scaling and Load Balancing via Cost Optimal MDP Solution }
%\author{Mark Shifrin, Erez Biton, Omer Gurewitz}%\\
%	Ben-Gurion University, Nokia} % \#1570010543}
\author{Mark Shifrin, Roy Mitrany, Erez Biton, Omer Gurewitz% <-this % stops a space
%\IEEEcompsocitemizethanks{\IEEEcompsocthanksitem M. Shell was with the Department
%of Electrical and Computer Engineering, Georgia Institute of Technology, Atlanta,
%GA, 30332.\protect\\
%% note need leading \protect in front of \\ to get a newline within \thanks as
%% \\ is fragile and will error, could use \hfil\break instead.
%E-mail: see http://www.michaelshell.org/contact.html
%\IEEEcompsocthanksitem J. Doe and J. Doe are with Anonymous University.}% <-this % stops a space
}
\maketitle
%--------------- table of contents ----------------------------------------
%
\begin{abstract}
We address a cost optimization problem faced by a user  who runs instances of applications in a remote cloud configuration constructed of multiple virtual machines (VMs). Each VM runs a single application instance which can execute tasks specific to that application.
Managing the VMs involves a sophisticated trade-off between cloud-related demands, which are expressed by the provisional costs of leased cloud resources, and exogenous cost demands expressed by service revenues that are typically bound to service level agreements (SLAs). %In particular, networking-based tasks' demands should comply with service level agreements (SLAs) while computational tasks have latency demands. %The  cost constraints apply at all cases.
The internal costs may include VM deployment$/$termination cost, and VM lease cost, per time unit. The exogenous costs refer to rewards accumulated due to the successfully accomplished tasks being run by each application instance. In the case where the SLA restricts performance to a certain load level at each VM, tasks incoming at VMs that reached that level are rejected. Rejections cause fines deducted against the rewards.
In addition, the performance level is also quantified, namely, by means of a delay cost, according to the average delay experienced by tasks.
%We address cost optimization problem of VM scaling and scheduling which considers both the exogenous and internal cost constraints.
%The user can monitor the state of all deployed VMs and to decide on scaling (deploying/terminating) a VM and load balancing (selecting a VM) upon each new task arrival.
Typical examples for specific applications which fall within this class of problems include handling of scientific worklflows and network functioning virtualization (NFV).  In the latter case the rejection penalties are particularly high.
%The existing VM scaling tools provide only limited instrumentation for cost optimization and are not designed to fully account for exogenous constraints.
%In particular, the scaling is handled by harnessing heuristic algorithms, while in order to effectively treat the aforementioned problem one needs to find and apply optimal scaling and load balancing policies.

%includes functionalities which are translated into activation and control of VMs according to stochastic demand of special networking tasks - virtual networking functions (VNFs). The successful accomplishment of such a task is generally associated with economical incentive while a failure is penalized, according to SLAs.
%For the combined problem this method  suboptimal heuristics.
We model this problem by cost-optimal load balancing to a queuing system with a flexible number of queues, where a queue (VM) can be deployed, can have a task directed to it and can be terminated.
We analyze the system by Markov decision process (MDP) and numerically solve it to find the optimal policy, which
captures the aforementioned costs and performance constraints.  %and includes decision making on queue deployment and termination, considers delay cost, cost associated with keeping a deployed VM (even if idle), task processing incentives which are modeled as scheduling reward and rejection fine.
Within this constrained framework,  we also investigate the impact of average VM deployment time.
We show that the optimal policy possesses decision thresholds which depend on several parameters.
We validate policies found by MDP, through directing an exogenous computational tasks flow, which is typical of image processing, to a set-up implemented on AWS. The policy which we propose here can be adopted by any cloud infrastructure. %within the existing AWS architecture.
 %within the existing AWS architecture.

%This research is a part of European Union Horizon 2020 Research and Innovation Programme SUPERFLUIDITY.   The conference version of this paper was published in~\cite{Shifrin2016}.

\end{abstract}
%\section{Introduction}

%To summarize our contribution, we presetn
%We also able to answer the question what is the average number of VMs that will be needed for the enterprise, for a given VNF scenario.

%%%%%%%%%%%%%%%%%%%%%%%
\section{Introduction}
%%%%%%%%%%%%%%%%%%%%%%%

 Cloud computing is a distributed paradigm which enables remote usage of resources (such as servers, storage, applications and information), by a vast user population distributed over the world. These cloud resources are typically allocated in the form of virtual machines (VMs) which are assigned to cloud users (typically applications) on a rental basis on demand \footnote{Even though our method applies both for public and private clouds, for clarity we will mostly refer to public cloud}. Obviously, to support the enormous number of applications utilizing the cloud and the large variability of incoming requests, a dynamic resource allocation mechanism is essential. Such resource allocation mechanisms are essential both by cloud providers who provide the infrastructure and the consumers, the applications. The objectives and constraints of either are not necessarily aligned and sometimes are even opposed. Specifically, on the one hand the cloud provider needs to place the VMs on physical machines taking into account various constraints such as determining which physical machine best suits each VM's requirements, spreading VMs over multiple physical machines for fault tolerance. Moreover, it needs to overcome failures, and accommodate bursty VM requests trying to admit all incoming requests, to ensure load balancing while cutting expenses, increasing revenue, etc. On the other hand, the user utilizing the cloud (an application) is charged based on the resources (VMs) it retains. Its incentive is to lower costs, hence to utilize as few VM resources as possible. On the other hand, the number of requests from the application’s consumers, handled by a VM hosting the application can affect performance (e.g., many requests, higher latency). Accordingly, the application aspires to deploy a dynamic resource allocation mechanism which adopts the number of leased VMs to the current demands. Due to its high importance and its impact on bringing cloud computing to more domains in our daily life, both aspects of resource allocation have drawn a lot of attention over the last few years both by the industry and by the academy. Notably, for most objectives the resource allocation problem is an NP-hard optimization problem (see, e.g.,~\cite{arfeen2011framework,li2011cloud} for the problem hardness discussion). Throughout this paper we concentrate on the application’s mechanism. We assume that the cloud providers are abundant with resources which they lease on request, and charge the application based on various pre-agreed parameters, such as the cost per VM, installing or removing VMs, etc.

As  mentioned, the trade-off between performance and number of VMs an application acquires, necessitates a dynamic resource allocation mechanism which increases or decreases the resources, in the form of VMs, to fulfill the elastic demands of the application. Such adaptive resource allocations are typically attained via two  mechanisms, load balancing and scaling. Load balancing is responsible for distributing the application traffic-load or application requests between the leased VMs, such that no VM is overloaded while others are less congested. Scaling refers to the procedure of adding (scale-out) or releasing (scale-in) VMs according to the demand \footnote{Scaling also refers to the procedure of adding or removing resources such as computing power, memory, storage, or network capabilities, to an existing machine (termed Scale-up and Scale-down, respectively). This procedure is not common hence will not be addressed throughout this paper.}. Clearly, the scaling mechanism which dynamically adds or removes resources (VMs) needs to ensure that the application’s clients will be satisfied and if there exists a Service Level Agreement (SLA) between the application and its consumers, it is kept. Accordingly, the dynamic resource allocation mechanism should balance financial considerations (e.g., adding or maintaining unnecessary VMs) and user satisfaction (e.g., performance).

Common application-resource-allocation mechanisms utilize an on demand approach, aiming at maximizing resource utilization by means of avoiding both overload and underload of VMs while maximizing users’ performance, typically consent with users’ SLA. However, it is clear that without pre-planning and, specifically, without trying to predict future demands, the resources maintained by an application can be satisfactory on the short term, yet badly utilized on the long run. Furthermore, due to its commercial nature, there are other perspectives, besides the complying with the SLAs, that should be considered. For example, due to the cost-effect reasoning, a user can prevent adding a VM (scaling out), thus compromising performance, or even risk a rejected task allocation request, or, on the contrary, bear the cost of multiple active yet underutilized VMs, in order to avoid future performance degradation and/or a rejection. In addition, a user can decide to migrate tasks from one VM to others in order to terminate a VM and save the expense or not to divert traffic to a VM in order to release it in the near future. Typically, the cloud providers implement application programming interfaces (APIs) for applications to set up several predefined parameters (thresholds) for execution scale out/in. %For example, to define that when CPU utilization is above a threshold scale out or if it is below a predefined threshold scale in (***FIXME OMER***).
However, a cloud user, who deploys an application(s) for clients, and reduces costs by controlling its own scaling and load balancing mechanisms, needs to take all these considerations into account, dynamically and automatically, transparently to the clients using the application, and without altering the application with respect to each action taken.

To this end, the set of costs incurred by the application can be conceptually divided into two types: the \textit{provisional costs} (PC) and \textit{service  revenues} (SR). PC refer to the payments incurred by the cloud provider, i.e., the costs that the cloud provider charges the cloud hosts (the applications) for utilizing its resources. These costs are typically per resource utilization (e.g., the number of VMs the application utilizes per unit time, the cost for acquiring or releasing a VM, etc.). The SR refer to the costs which are \textit{exogenous} to the cloud and stem from the SLA (with exogenous clients) and the performance level measured by the cloud user, i.e., the enterprise through its cloud exploitation. Specifically, the PC are accumulated continuously according to the retained resources (VMs) which are controlled dynamically via the application \textit{scaling} mechanism. Namely, an enterprise which runs the application decides on scale out/in operations which increase/decrease (i.e., deploy/terminate) the number of VMs, respectively. Obviously, utilizing scale up and down mechanisms which can upgrade (add resources to) or downgrade (remove resources from) existing VMs, can also effect the PC, however these procedures which are not broadly supported (e.g., they are not supported by Amazon Web Services (AWS) which we exploit throughout this paper) and are not addressed in this paper. On the other hand, scale up and down mechanisms do affect SR by various factors. The foremost factor is the performance which is directly affected by load on each VM, which is affected by both the scaling and load balancing mechanisms. Note that the effect on performance is not necessarily a simple function of the number of employed VMs or the number of users utilizing the application. For example, the expected delay experienced by a video streamer request, definitely relies on the load of the VM which addresses
this request, however it relies on several other factors (such as the negative impact of other requests being concurrently served on the same VM ) and is not necessarily a linear function of the load. In our mathematical formulation, we address this generality.

Throughout this paper, we will assume that the PC are typically set according to a price-list which is agreed upon on advance (note that in private clouds it is possible to quantify the costs of utilizing resources as well). As for the SR, it is expressed via rewards or fines, associated with successfully processed or rejected tasks, respectively. We assume that these costs are negotiable between the application and its users, yet this is determined a-priori (e.g., via SLA). In order to materialize the service quality associated with the performance (e.g., delay), we translate it to \textit{performance costs} which stand for additional fines inflicted on the enterprise, thus encouraging faster processing.
%(**** Mark I do not understand this next paragraph: Note that while the measured performance depends on the internal cloud parameters, the costs incurred by the experienced delay refer to the SR. This is due to the fact that the actual value of the delay is dictated by the SLA between the enterprise and the users whose applications it manages.
Typically, the calculation of these costs is implemented by means of pre-negotiated delay cost factors (DCF) which appropriately scale the delay function. The DCF value versus the experienced delay provide an additional complexity to the problem. %***).
%(*** I think this should be part of the model: We will assume that the enterprise will know the exact status of its VM configuration at all times\footnote{In reality this knowledge is easily acquired by monitoring tools, e.g., CloudWatch in AWS.} ****).
%(*** To this end,
We term by \textit{cost parameters } (CP) the fixed set of all costs VM deployment/termination cost, DCF, reward, rejection fine and VM holding cost. % ***???**).

In this paper, we explore and formulate the enterprise’s objective of optimizing its revenues under the set of costs. Specifically, we devise a decision maker (DM) which implements optimal scaling, load balancing and admittance control, based on the predefined CP and the current known system state.

Scaling and load balancing in the cloud has drawn a lot of attention over the last few years both by the academy and by the industry; many studies have addressed the problem of dynamic scaling in the cloud. The review on cloud auto-scaling~\cite{lorido2014review} would categorize our work within a class of reinforcement learning (RL-based) scaling techniques, as in, e.g.,~\cite{tesauro2006hybrid} which utilizes Q-learning for resource allocation in data center. Yet, works mentioned wherein do not account for load balancing and do not capture the cost trade-offs we address.
Another class of works (~\cite{chieu2009dynamic}) concentrate on architectural novelties, while the scaling is governed by heuristically set thresholds. The recent survey in~\cite{xu2017survey} on load balancing reasons the NP-hardness of the general LB problem in cloud and categorizes the existing algorithms into heuristic and meta-heuristic. However, we argue that for the carefully defined assumptions, an optimal solution can be presented, without sacrificing the major part of the generality. Hence, this paper is the first to give a mathematical formulation from the perspective of the enterprise's long-run cost optimization, while capturing the trade-off structure in this degree of generality. The solution we provide directly leads to the practical methodology, to be implemented within existing or yet-to-come scaling and load balancing architectures.

To this end, we model the problem by a \textit{queuing system with a dynamic yet limited number of queues}. %; each queue stands for VM running VNF instances.
The controllable queuing system represents the dynamic VM deployment and tasks balancing problem which we treat by introducing a \textit{stochastic control model based on Markov decision process (MDP)} formulation. The solution to the MDP provides the optimal policy.
%We assume that tasks arrive with constant average rate.

While the set of costs described above implies no trivial policy could exist, some of the parameters have contradicting impacts which may dictate certain properties of the policy.
For example, in the case where the delay cost function sharply increases with the load, the DM will attempt to deploy as many VMs as possible.
On the other hand, in the case where keep-alive cost is comparatively high, the DM may prioritize minimization of the active VMs number.
A distinctive impact has a \textit{VM deployment time}, which we separately explore.

%formulated as a Markov Decision Process.

%One of the objectives of our model is to facilitate a fast scaling of VNF.
Once applied, the policy potentially reveals the optimal average number of active VMs and the corresponding optimal average number of tasks in the queue associated with the VMs.
This allows the enterprise to assess the demands and to plan ahead.
Moreover, the solution to the MDP is expressed by \textit{value function} which indicates what are the costs and revenues the enterprise will receive in the long run, for a given scenario along with its set of parameters.

To summarize, the contribution of this paper is as follows:
\begin{itemize}
	\item We formulate load balancing of task arrivals and VM scaling problem by stochastic model solvable by MDP.
	\item The MDP is numerically solved providing the optimal policy which captures the various costs constraint.
	\item A detailed research of the value function structure is performed and the insights of the structure of the optimal policies are brought.
	\item We separately investigate the impact of the VM deployment time.
	\item We built a simulation set-up on the basis of AWS platform and apply the policy retrieved from MDP solution.
\end{itemize}
Note that this set of constrains we considered was more elaborate than AWS setting allows.
Hence, while we used AWS for the exemplifying purposes, the method we present in this paper is generic, and will fit any cloud based platform. The only needed input for the MDP consists of system parameters (e.g. average VM deployment time, average task execution time). The implementation methodology, which is suggested to be harnessed by enterprises, may be expressed by a closed loop, as it is schematically depicted  in Figure~\ref{figCL}. %for the schematic depiction of the closed loop combined with a cloud platform.
\begin{wrapfigure}{R}{0.5\textwidth}
	\begin{center}
		%\begin{align*}
		\centerline{\includegraphics[width=15em]{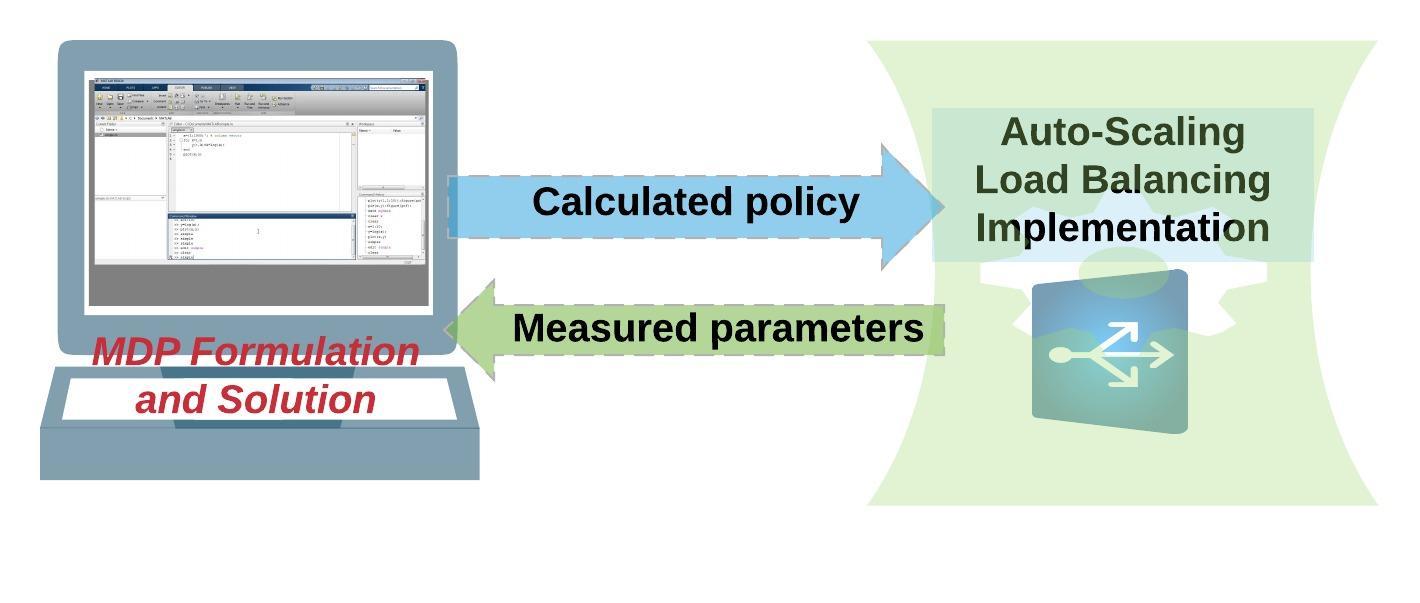}}
		%\end{align*}
		\vspace{-10pt}
		\caption{\sl\small
			Policy calculation closed loop. The scaling time statistics are firstly measured for an arbitrary policy and then used for the optimal policy calculation. At least one additional measurement is needed to adjust to the impact of the MDP driven policy. }\label{figCL}
	\end{center}
\end{wrapfigure}

The rest of the paper is organized as follows. The next section provides general background on VM control in cloud computing, by exemplification via NFV and scientific workload scenarios. This part might be familiar to the experts in cloud computing. We provide system description in section~\ref{sec:syst}. %Section~\ref{sec:mdp} deals with mathematical formulation of the MDP.
Section~\ref{sec:mdp} provides the MDP formulation. The study of the optimal policy and of the impact of different parameters is discussed in Section~\ref{sec:Ores}. Section~\ref{sec:AWS} provides the review of our AWS-based implementation and demonstrates the results of applying the optimal policy on AWS. Section~\ref{sec:related}  gives the related work and concludes this paper. 
\section{VM cloud-based control background}\label{dec:prel}
We provide preliminaries and bring the related examples which demonstrate the problem of VM control in practice. We rely on two scenarios. The first one is driven from NFV paradigm while the second one is from offloading massive scientific computation workloads. Both examples are characterized by the tension raised by an enterprise's aim to maximize the revenues accumulated in the process of successful task execution, on the one hand, and keeping the cloud expenditures as low as possible, on the other hand. 
The breakthrough of the NFV concept gave rise to novel demands triggered by the goal to cost-effectively run persistent networking functionalities. 
Fulfilling the cost optimality demands, in NFV context, implies 
binding VM deployment with solutions to variety of networking aspects. 
%\begin{comment}
Fo the clarity, we name three possible interacting sides as follows: 1) The \textit{users}, who supply the demands, 2) The \textit{operator} which is responsible for all control decisions and 3) The \textit{cloud provider} which supplies the VMs.   
Once a VM is deployed, it is uniquely associated with a specific task type, which is related to a specific virtual networking function (VNF). 
%\end{comment}
Note that two possible configurations of deployment are possible: A) where all VMs are deployed over the remote cloud. In this case all HW-related costs are paid by operator to the cloud provider and B) The cloud is private and is actually owned by the operator. The costs are associated with the operator's activities for handling the owned VMs. 

A deployed VM, which is normally started with an image containing the desired VNF, is disposed to handle a flow of latency sensitive tasks (according to the SLA) which are typically directed to it by a separately (on-premises or externally) implemented orchestrator which combines in itself a decision maker (DM) which is responsible for scaling. In addition, there is an virtualized entity which is responsible for the load balancing. (It may be deployed separately from an orchestrator.)
%Every VM which has a VNF instantiated on it, acts as a part of general scheme which is exemplified in $(a)$ to $(e)$ stages in Figure~\ref{fig22}. Henceforth, in this paper, as far as NFV paradigm is concerned we focus on parts $(a),(e),(f)$ of the scheme. 
We  assume that the functionality type of VNF of interest is known and fixed, and that there is a constant influx of tasks to be served by the corresponding VMs deployment, where all such VMs  have instances of the corresponding VNF implemented on them. 
%An enterprise, which decides to virtualize any of their network services or effectively offload their computational demands needs first to \textit{deploy} ("\textit{build}") the corresponding service, in order to be able to \textit{run} the tasks of that type. 
%The deployment process involves having leased a VM and loading on it the corresponding software. 

The procedure of deployment takes time which depends on various factors which include the type of VM, availability, booting time,  and deployment of the image which includes the application.
We account for a one-time deployment and termination cost applied for each VM. In the case these costs are not applied at a specific cloud provider, they are just assumed to be zero in the model formulation. %Hence, having idle resources is undesired. 
%On the other hand, the delay involved with the deployment or lack of space for the new tasks can cause some tasks to be rejected from running, thus inflicting a profit loss to the enterprise.
%Note that in some cases, the \textit{termination} (sometimes referred to as"\textit{destroy}") operation, i.e., the process of releasing VMs can 
%also incur a cost. 
\begin{comment}

%In order to show the connection of the objectives, we pursue in this paper, to the general scheme, we use the NFV example. 
\begin{wrapfigure}{R}{0.4\textwidth}
%	\begin{center}
		%\begin{align*}
		\vspace*{-20pt}
		{\includegraphics[width=15em]{VNFgeneral2}}
		%\end{align*}
		\caption{\sl\footnotesize
			NFV mission chart possible phases:  
			$(a)$ Multi-user enterprise which runs multiple networking related tasks.
			$(b)$ The internal engine decides about internal priorities and NF types which are about to be virtualized on the cloud. 
			$(c)$ VNF optimal placement engine, which conducts VNF placement with accordance to service precedence and user types and locations. 
			$(d)$ Authentication of the outgoing tasks and permissions settings. This block should be coordinated with cloud security groups settings.
			$(e)$ The DM which decides on deployment and displacement of new VMs and performs load balancing. (Load balancing may be implemented as a separate automated part, based on the cloud provider features.)
			%\end{enumerate} 
		}\label{fig22}
%	\end{center}
\end{wrapfigure}
\end{comment}
Once deployed, the operator pays per time of having the deployed VMs, regardless of the load. This cost is charged per unit of VM's leased time. In AWS, for example, the payment is normally applied per hour of a usage. Alternatively, in a private cloud owned by the operator, the payments are related to costs of not leasing those VMs for other revenue making applications.
The tasks are directed to VMs upon a connection establishment via the load balancer (LB). The LB might be separately defined and deployed by an operator who sets its configuration, yet physical details might be left transparent\footnote{The alternative, where the operator would not use cloud provider controlling tools is also possible. In this case she will implement her own load balancing and orchestrating SW and will use it for addressing VMs. }.
For simplicity, our model does not account for the costs associated with the deployment of the load balancing and orchestrating machines. Operator communicates VMs via the load balancer, which can monitor the VMs at all times and the orchestrator which is responsible for the scaling. The load balancing of the tasks and scaling decisions are made according to the load at VMs. In this work, we assume that monitoring mechanisms are implemented and provided by the cloud provider. In AWS, for example, it can be done by CloudWatch and alarms, which prove to be very effective and provide versatile controlling flexibility.  
% If less VMs then less server = less money, also less VMs = more VMs for other application that can bring revenues
%Maybe refer to that in the private cloud case, the cloud provider is from a different devision from the VNF operator, and internal funds are transferred between the devisions

We express the delay cost function via weighting an average delay by a constant termed delay cost factor. This constant fits the incentives associated with  quality of service demands and the SLA's on the side of the enterprise to the performance level it could achieve on their system deployed on the cloud. Note that we allow the delay cost function to non-linearly  grow with the load. That is, even for the parallel processing, the delay can increase faster than linear.
Note that this phenomenon is especially obvious for NFV use cases. For example, at network-aware application, when VM share both CPU and transmission bit-rate. Note that the delay is not directly fined by the cloud, but its implications are merely translated into economical values.

%In the two aforementioned examples, 
%Yet, it is less important in quality-minded treatment of computational task flows. The architecture of transferring of computationally intensive scientific missions is comparatively simple. Once a task had been dispatched to a VM, it occupies a slot of resources till its execution is finished. Then, the calculation results are transferred to the enterprise using the same orchestrator.
The area of offloading scientific workflows gained a recent promotion in the sense of various heavy computational missions, e.g., long-time training of large-scale multilayer neural networks, treatment of huge database queries within Big Data solutions, etc. The specific works dedicated to such implementation approaches are listed in section~\ref{sec:related}.
%Consider an enterprise which needs to run networking tasks of a certain type. Namely, we assume all VMs are similar and able to run similar VNF tasks (e.g., flows a firewall handles), or tasks associated with scientific flows of similar type. 
In contrary to the NFV scenarios, where in most case rejections are intolerable, a strict limit to the number of allowed tasks at each VM can be set. 

The considerations and constraints described in the both examples suggest a trade-offs between delay costs, task rejection thresholds and VMs scaling thresholds which can not be easily predicted, and thus have to be formally analyzed.
Even as in NFV case, where the rejection fines are set to be particularly high, there is no clear intuition how to choose VMs scaling thresholds. Hence, we explore how the delay function impacts the decision policy.

%Hence, the decision whether to scale out or in, typically performed by the VNF manager (VNFM), remains an important problem that needs to be addressed, especially in dynamic scenarios. 

We argue that  modeling of the problem by \textit{queuing system} with a dynamic yet limited number of queues, where each queue stands for VM running application instances, is a natural approach. 
The solution can be readily provided for any given maximal number of VMs, yet due to the physical constraints of the cloud provider and economical limitations of any enterprise, we assume the number is finite. This assumption, likewise, allows for a simpler model formulation and analysis.
%At departure from a queue (i.e., running of VNF task ends) which has been left empty, DM decides whether to keep that queue alive or to destroy it.   
%In conjunction with the scaling decisions, we are also facing a load balancing challenge, steering the traffic flows to the different VMs and balancing the load between them. In this study, 
%Hence, we tackle both the scaling decision and the scheduling, i.e. the  load balancing strategy as a single problem. 
In the most simple scenario,  load balancing will merely amount to having equally loaded VMs. The trade-off in this case means having a high average load with a small number of VMs versus having a low average load with a larger number of deployed VMs. This is the scenario we test on AWS setup, as is explained in section~\ref{sec:AWS}. However, having in mind deployment time and cost, this trade-off still represents a significant challenge as the optimal policy derivation is not straightforward.

We additionally assume that the VM placement problem and authentication issues are independently solved prior to the scheduling, and that the solutions are static or have no effect on scheduling-related parameters (e.g. VM deployment time.)
Henceforth we focus on the cost-optimal VM scaling and the load balancing challenges which are relevant for a given specific type of tasks. In order to circumvent privacy-related issues, we also assume that all tasks are coming from the same user identity. Extending to several users or/and to several task types is straightforward and merely converges to solving  several independent problems, where only minor adjustments to the setting are needed.

Note that NFV and offloading workflows are only a portion of the possible scenarios that can be associated  with the  described system. We believe that  NFV is a best exemplifying candidate both because of its global nature and of the fact that it is clearly associated with a persistent long-term task flows. Hence it constitutes an obvious yet open and important problem of long-run cost optimization. Any other entity with persistent flows of tasks can be considered, provided it introduces the appropriate translations of task completion successes and failures to the rewards and fines.

For the sake of exemplification, we will use NFV terminology whenever further detailed exemplification is needed. 

\section{Formal System Definition}\label{sec:syst}

We assume that the available cloud resources (e.g., NFV infrastructure) can host a finite yet flexible number of VMs. %In the NFV context, for simplicity, we assume that an instance of a VNF is deployed in one VM, with two instances deployed in two VMs and so forth. 
%That is, in a scale out operation, adding a VM translates to adding another instance of the VNF. 
We further assume, that a load balancer (LB) is deployed and can handle all traffic demands irrespective of the number of VMs, see Figure~\ref{fig1} for the schematic presentation. %which corresponds to the detailed depiction of parts $(a),(e),(f)$ in Figure~\ref{fig22}. %Note again that the description in Figure~\ref{fig1} conforms to the general problem of handling computational or other processable tasks flow through the cloud, rather than being limited to the NFV scenario. For example, AWS provides an infrastracture which includes an elastic load balancer (ELB) which acts over a predefined scaling group.

Our theoretical model is based on a Markovian assumption. That is, the service demand is modeled by a Poisson process of arriving tasks with average rate $\lambda$.
Upon each task arrival event, (service request), the decision making (orchestrator) decides whether to instantiate a new queue (VM) to handle the demand, as long as the maximal number of active queues (VMs) is not reached. In addition, it directs the task towards a load balancer, which balances tasks across active queues. Once service is ended (i.e., task departure) and a VM is left idle (the queue is left empty), the orchestrator may keep the VM or “destroy” it.
%Upon each task arrival event, the Decision Maker (DM) decides to which VM the arriving task should be directed. The VM deployment decisions are made at arrival times, as long as the maximal number of active VMs is not reached. 
%At departure from a queue  which has been left empty, DM decides whether to keep that queue alive or to terminate it.   
In what follows, we will use naming Decision Maker (DM) in order to refer to the operator and queue in order to refer to a VM.

The maximal number of running task on a single VM is limited by the borderline number, above it the performance degrades below the minimal quality of service and, hence, should never be exceeded.

The DM aims to find a policy which  maximizes the total income in the long run.
We assume exponentially distributed service time and a time it takes from the moment of VM deployment decision till the moment it is fully deployed. While the arrival part of the Poisson assumption is rather natural, the same assumption about the service and deployment part mean that we impose an approximation of the service times and on VM deployment times. However, the drawback of this approximation proved to be non-significant, as show our results in Section~\ref{sec:AWS}.
Moreover,  to account for general (but known) service times, it is merely needed to extend our MDP model to a Semi-Markov Decision Process (SMDP) model, an effort that is purely technical. As the objective of this work is to present a general methodical paradigm, we leave the extension to SMDP out of the scope of this paper.
Accordingly, incoming tasks are scheduled to one of the active queues or rejected from service, according to  the load balancing policy.
Each  VM  can handle services in a parallel manner. 
Namely, the total processing rate is equally shared between all tasks currently running in the queue. Hence, tasks never wait for a full completion of previously arrived tasks but are rather processed in parallel. The VM's limited resources allow processing of a limited number of services, denoted by $B$. The parameter $B$ is calculated based on the service SLAs. We omit these detailed calculations and assume $B$ is given. For simplicity, we assume all VMs are identical, hence characterized by equal $B$.
We assume that the  exponential service times have maximal average rate $\mu$. We assume that task processing initiation at each VM has no time overhead. However, VM deployment time is significant. For analytical simplicity we also assume it is exponentially distributed with average rate $\zeta$. 
%exponentially distributed rates ***FIXME*** is exponentially distributed with a rate equal to the sum rates
Each VM is modeled by a queue with a buffer size $B$, having up to $B$ servers with a total processing rate equal to $\mu$.
Since the minimum of exponentially distributed rates is equal to their sum, the total processing rate is always equal to $\mu$. 
%Hence, 
While this model reflects the ability of VMs to provide concurrent resource sharing, it can be easily modified for the FIFO service, with only minor changes. 
Note that for the correspondence with the mathematical model, which will follow, we use the terms VM and queue interchangeably.

The Service revenues (SR) are primarily composed of rewards for admitted tasks and fines for rejected tasks. We assume that an admitted task gives a fixed reward, while a rejected task incurs a fine, denoted by $\{r,f\}$, respectively. 
%Each queue can be deployed and destroyed dynamically, according to the demand.
Intuitively, the number of queues may infinitely grow for some sets of parameters, as long as the total cost is minimized. For example, in the case where $\lambda\gg\mu$, and rejected tasks incur high fines, it always profitable to have as many VMs deployed as needed in order to avoid tasks rejections. On the other hand, the number of available queues is expected to be bounded by both general system limitations and economic considerations.
We will naturally assume henceforth that the number of active queues is limited. % given $\bq$.
%We also assume that all tasks are processed in FIFO order, with equal rate at all queues (VMs). For simplicity, we assume the processing times are exponential, yet, general SMDP model formulation allows any known service time distribution. 
In case that there is no active VM, i.e., immediately after the VNF on-boarding and before its VM deployment  or in case that all active VMs are overloaded and cannot accommodate new tasks an incoming task would be rejected. However, for the hypothetical configuration where the deployment time is instantaneous, rejection could be avoided by an immediate VM deployment. This configuration merely has theoretical value and is unpractical. Hence we will assume for our validation setup that the deployment time is not negligible at all times. Furthermore our equations account for the deployment times. 
%We model the described system by finite yet \textit{flexible} number of finite queues, which can be deployed and destroyed. Each such queue stands for NFVI point of presence. 
%The total demand of VNF instances originated by the enterprise is modeled by Poisson process of arriving tasks. The incoming tasks are scheduled into one of the \textit{active} queues, according to some policy $\pi$. 
\subsection*{Set of the queues}
We define next vector of queues, which describes both the state of all queues and the number of tasks in each one of them. Hence, such a vector uniquely reflects the system state. We will need the following definition for this purpose.
Denote the set $\mq=\{-2,-1,0,\cdots,B\}$. 
The vector of all queues at time $t$ is denoted by $\bq(t)$. The maximal number of VMs is given by the size of this vector, namely $|\bq|=\bn$, for some system-related $\bn$ which was precalculated and fixed. %, where each queue is denoted by $q_i(t)$. 
The state-space (st.-sp.) defines all possible states the VMs (queues) could possibly have. Denote it by the following
\[
\Q=\mq^\bn\;\;,\bq(t)\in\Q\;
\]
The components of $\bq(t)$ correspond to the number of tasks in queues and are denoted by $q_i(t)$, where $i$ is the general indexing of the queues. To this end, $q_i\in\mq$. 
The state "$-2$" means the queue is inactive and state "$-1$" means the queue is being currently deployed. Note that this is different from state $0$ which means the queue is active and deployed, but empty.
(We will omit the time notation in cases where the current time is of no importance to the analysis.) While for simplicity we assumed $\mu_i=\mu$, the extension to the system with different processing rates is straightforward, at expense of the appropriate st.-sp augmentation. %Each queue serves tasks with equal rate with average equal $\mu$.
%\subsection*{Set of states for a queue - the State-Space}
%The queues are limited by $-2\leq q_i(t)\leq B$, $i\in\{1,\cdots,|\bq|\}$,
%The vector $\bq(t)$ forms the \textit{state-space} of the MDP at time $t$. 
The number of all states is denoted by $\rvert\Q\rvert$. We will use the definition of the state space $\Q$ in the next section for the formal MDP formulation.
\subsection*{Cost structure}
We now formally define the cost parameters which can be logically divided to the provisional costs (PC) and the service revenues (SR).
%For simplicity, we assume all queues are identical with service rate equal to $\mu$. 
%The reward for each admitted task and the fine for rejected task are denoted by couple $\{r,f\}$.
%There is fixed cost for queue deployment and destroy and fixed maintenance cost per time unit for each active queue. 
%Prior to the definition of the costs set, we note that not all the cost types which will follow are directly charged by all cloud resource providers. However, some of the corresponding actions require delay, which we assume was properly expressed by a suitable cost. %That is, the delay is not directly fined by the cloud, but its implications are merely translated into economical values.
We define first the structure of the delay cost. Denote by $h$ the delay cost per unit of processing time, 
associated with a number of hosted tasks. That is, $h$ serves as a delay cost factor (DCF) responsible for the part of the SR which associate the SLA with the performance level. At $i$th queue, at time $t$,  
cost equal to $h_i(t)$ per unit of time is inflicted. %given by $q_i(t)$, is denoted by . %We address two delay models. 
%In the first model, linear delay is given by 
%\begin{equation}
%h_i(t)=q_i(t)h \label{1}
%\end{equation}
The delay cost model is represented by function which increases with the number of hosted tasks at VM.
This structure reflects a load increase (and, consequently, the overall queuing time) with the number of residing tasks at a VM.
In particular we express this load impact by using an increasing function $\eta$ as follows,   
\begin{equation}
h_i(t)=q_i(t)\eta(q_i(t))h, \label{1b}
\end{equation}
Where $\eta\geq1$ for all possible $q_i$ values and is positive increasing with $q_i$. 
Namely, more busy VMs perform with higher delay. Hence, this cost structure penalizes for having busy VMs. Therefore, this formulation allows to capture parallel processing, such that the delay cost of \textit{all} running tasks is appropriately weighted by their quantity at a VM.
Note that if we assume $\eta(i)=1$ for all $i$, the cost will degenerate to the simple linear model.
The cumulative delay cost till time mark $t$ is merely given by
\[
H_i(t)=\int_0^{t}q_i(t)\eta(q_i(t))h\:dt
\]
The DCF, together with earlier defined rewards and fines $\{r,f\}$ accomplish the definition of SR.
As for cost parameters related to PC,
the deployment cost $\beta$ is applied each time a queue is activated. The termination cost $\psi$ is applied each time a queue is terminated.
%To 
%Heuristically, the difference function $\alpha$ is composed of holding cost and delay cost. The holding cost cancels out since $a$ and $b$ have same number of empty queues. Similar assertion holds for $a'$ and $b'$, $a''$ and $b''$. However, the delay cost difference in states $a'$ and $b'$ is lower if compared to $a''$ and $b''$, because additional packets in the same queue cause higher delay, while having first packet in a new queue. Generally, 
Observe that the total delay cost is the lowest when tasks are equally dispersed over queues. %In particular, observe the delay structure in~\eqref{1b}.
%We now introduce queue maintenance costs. %by defining the sets as follows.
%The \textit{queue operation cost set} includes deployment, destroy and maintenance cost per unit of time and given by the triplet $O=\{b,d,k\}$.
%The first to parameters are values which are instantaneously augmented to the cost upon queue activation (i.e. VM deployment) and deactivation (i.e VM destroy).
%The third parameter $k$ stands for the cost per unit of time paid for any active queue.
Once a VM is empty no delay cost is applied, however the keep-alive cost for having a deployed VM is always charged, even if the VM was idle.
Denote this cost by $\kappa$ per a time unit. This cost is also a part of PC.
Figure~\ref{fig1} depicts a physical enterprise which constantly offloads tasks with average rate $\lambda$ to his virtual private space which contains up to $3$ VMs, where minimal tolerable service rate is given by $\frac{\mu}{4}$.

\begin{wrapfigure}{R}{0.55\textwidth}
%\begin{center}
%\begin{align*}
%\centerline{\includegraphics[width=26em]{sys2}}
\includegraphics[width=26em]{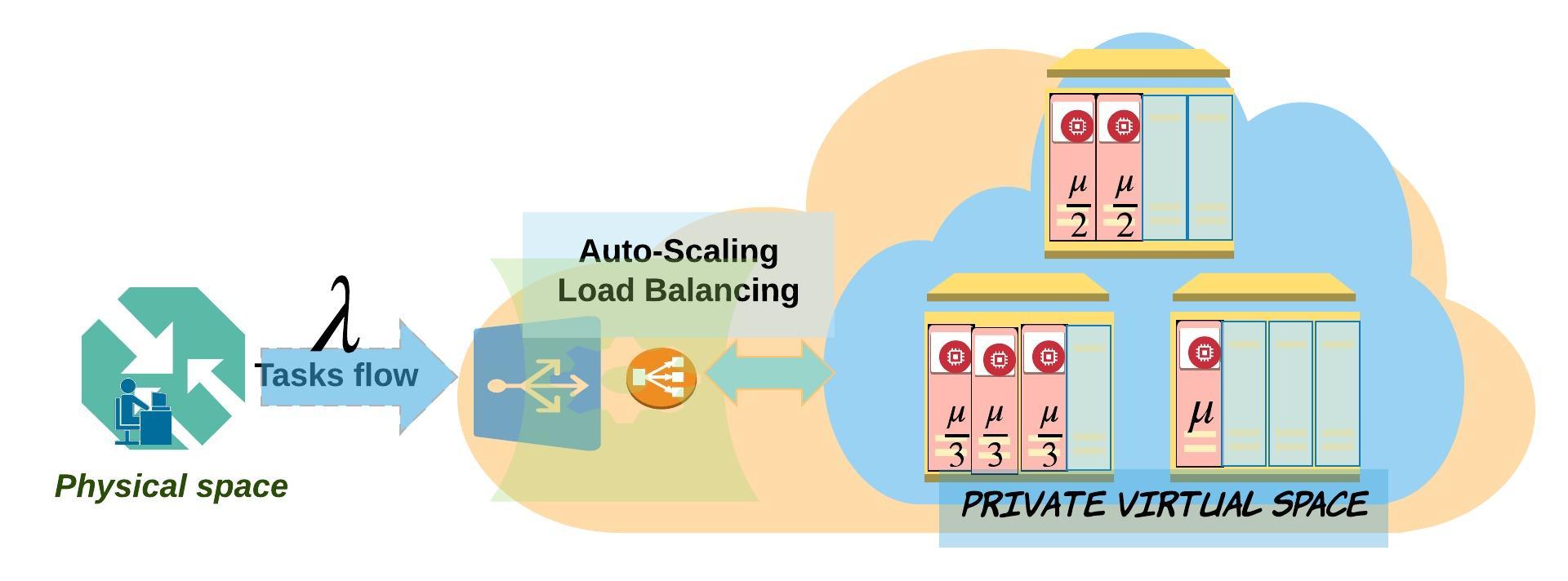}
%\end{align*}
\vspace*{-10pt}
\caption{\sl\footnotesize
System scheme, containing up to $3$ VMs with maximal capacity of $4$ tasks}\label{fig1}
%\end{center}
\end{wrapfigure}

We are now disposed to formulate the MDP.
%\subsection*{MDP statement}
%Simulations, policy analysis.
\section{MDP formulation}\label{sec:mdp}
%We first summarize all system parameters.
%Denote VNF instances arrivals as Poisson arrival process with average inter-arrival times given by $\lambda$. 

%The cost of queue deployment (resp. destroy) is denoted by $b$ (resp. $d$).
%The cost of the active queue is given by $h$ per unit of time.

%At each task arrival, the task will be scheduled into one of the queues, according to some policy $\pi$.
%Generally, the scheduling policy may be predefined, or be part of the optimization objective. 
We define the action space first, denote it by $\A$. Denote action $\ba\in\A$ at time $t$ as
\begin{footnotesize}
	{\footnotesize
\[\ba=\{\bu,\bb,\bd\}\]
\normalsize}
\end{footnotesize}
In particular, denote the \textit{load balancing vector} $\bu$ of length $\bn$, such that
\begin{footnotesize}
	{\footnotesize
\[
u_i(t)\in\{0,1\},\;\;0\leq\sum_i^\bn u_i(t)\leq1
\] 
\normalsize}
\end{footnotesize}
The sum is equal to $0$ in the case the decision was a rejection. Otherwise, the scheduling decision at $t$ is expressed by scheduling into queue $i$, hence 
\[\exists i,i\in[1,\cdots,\bq]\mid u_i(t)=1,u_{j\neq i}=0\]
The {queue activation policy} is applied at each task arrival and is described by \textit{build action vector} 
\[\bb(t)=\{b_i(t)\}, \;b_i(t)\in\{1,0\},\]
where $i$ indexes the queue and the possible values for $b_i$ stand for "deploy", and "not deploy", respectively.  
The \textit{queue termination policy} is applied at each departure event and is described by \textit{terminate action vector} 
\[\bd(t)=\{d_i(t)\}, \;d_i(t)\in\{1,0\}.\]

In what follows we will deal with counting processes of the general form:
\begin{footnotesize}
	{\footnotesize
\[
V(t)=\sup\{m\;; \sum^m_{i=0}v(i)\leq t\}, %\label{eq:sup}
\]
\normalsize}
\end{footnotesize}
where $v(i)$ is the time between an increment (e.g., arrival time, service time) $i-1$ and $i$, for some process $V$. 
Denote the arrivals counting process as $A(t)$ and task completion counting processes as $\{D_i(t)\}$, where $i$ indexes the queues. %Process counting rejected tasks denoted by $R(t)$.
Define the following indicator functions:
\begin{footnotesize}
{\footnotesize
\begin{definition}[Queue indicators]
For $1\leq i\leq\bn$
\vspace*{0pt}
\begin{align*}
%\begin{itemize}
&\text{Inactive queue :}  
&\bI_i^{\bi}(q)=1\;\;\; &\text{ iff } &\;\;q_i=-2\\
&\text{Deploying queue :}  
&\bI_i^{\bd}(q)=1\;\;\; &\text{ iff } &\;\;q_i=-1\\
&\text{Empty but active queue (idle):}  
&\bI_i^{\bee}(q)=1\;\;\;  &\text{ iff }& \;\;q_i=0\\
&\text{Queue with exactly one task: } 
&\bI_i^{\bo}(q)=1\;\;\; & \text{ iff } &\;\;q_i=1\\
&\text{Full queue: } 
&\bI_i^{\bff}(q)=1\;\;\;  &\text{ iff }& \;\;q_i=B\\
&\text{None of the above (denoted as normal): } 
& \bI_i^{\bn}(q)=1\;\;\;  &\text{iff } &\;\;2\leq q_i<B
%\end{itemize}}
\end{align*}
\end{definition}
\normalsize}
\end{footnotesize}
%We will omit the reference to the state $q(t)$ to which the indicators are applied where the state is clear from the context.
In most general form, these actions are allowed to be taken at arrival events and any departure events, that is, once the counting processes $A$ and $D_i$ increase. 
%The delay cost is added for all non-empty queues, according to the number of residing tasks. 
%Finally, the scheduling reward or fine is applied each arrival.

Define the infinite horizon discounted cost functional, discounted with discount factor $\gamma$, for policy $\pi$ as follows:
\begin{footnotesize}
	{\footnotesize
\begin{align}
& J^{\pi}=\int_0^\infty e^{-\gamma t}\cdot\no\\
&\Big[-\big(\bb(t)\cdot \beta+
\bd(t)\cdot \psi\big)\big(dA(t)+\sum_i^\bn dD_i(t)\big)\label{eqJ1}\\
&\;-\sum_i^\bn\big(h_i(t)+\kappa*(1-\bI_i^{\bi})\big)dt\label{eqJ2}\\
&\;\;+(\bu(t)\cdot r-f*(1-\sum_i^\bn \bu_i(t))dA(t)\Big], \label{eqJ3}
\end{align}  
\normalsize}
\end{footnotesize}
where $h_i(t)$ is substituted from Equation~\eqref{1b}. The cost can be divided into the components as follows. The first part of the display above, i.e.~\eqref{eqJ1}, stands for the queue deployment cost, denote it as $J^{\pi}_b$, and termination cost, denote it as $J^{\pi}_d$. The second part, i.e.~\eqref{eqJ2}, stands for queue holding cost, denote it as $J^{\pi}_h$, and delay cost, denote it as $J^{\pi}_\kappa$.
The third part, i.e.~\eqref{eqJ3}, stands for the cost associated with scheduling rewards, denote it as $J^{\pi}_r$, and the cost associated with rejection fines, denote it as $J^{\pi}_f$. That is, the cost is otherwise written by using the aforementioned components as follows, 
\[J^{\pi}=-J^{\pi}_b-J^{\pi}_d-J^{\pi}_h-J^{\pi}_\kappa+J^{\pi}_r-J^{\pi}_f.\]
%\begin{footnotesize}
%	{\footnotesize
%\[
%J^{\pi}=-J^{\pi}_b-J^{\pi}_d-J^{\pi}_h-J^{\pi}_\kappa+J^{\pi}_r-J^{\pi}_f
%\]
%\normalsize}
%\end{footnotesize}
The value function associated with initial state $q$ is given by
\begin{footnotesize}
	{\footnotesize
\[
V_{q}=\max_{\pi}J^\pi(q).
\]
\normalsize}
\end{footnotesize}
We now write the Bellman equation for a simplified and more realistic scenario
assuming that build operations can be only done at arrivals, while destroy operations can be only done at departures.
Denote by $e_i$ vector of length $\bn$ with value $1$ at $i$th coordinate and zeros in all other coordinates. 
The Bellman equation reads
\begin{footnotesize}
	{\footnotesize
\begin{align}
& V_q=\big[\sum_i^\bq\bI_i^{\bn}\mu_iV_{q-e_i}+\sum_i^\bq\bI_i^{\bff}\mu_iV_{q-e_i}
+\sum_i^\bq\bI_i^{\bo}\mu_i\max\{V_{q-e_i},V_{q-2e_i}-\psi\}+\sum_i^\bq\bI_i^{\bd}\zeta_iV_{q+e_i}\nonumber\\
&\;+\lambda\max\big\{\max_{\bb=\{0,1\}}\big[\max_{i,\bI_i^{\bff}=0,\bI_i^{\bd}=0}\{V_{q+e_i}-\bb\beta\Pi_j\bI_j^{\bi}+r\}\big],\max_{\bb=\{0,1\}}\big[V_{q}-f-\bb\beta\Pi_j\bI_j^{\bi}\big]\big\}
+C(q)\big]\delta_q, \label{eq:b}
\end{align}
\normalsize}
\end{footnotesize}
where the cost function $C(q)$ and normalization factor $\delta$   are calculated by
\begin{footnotesize}
	{\footnotesize
\begin{align}
&C(q)=\sum_i^\bq h_i+\kappa*(1-\bI_i^{\bi})(1-\bI_i^{\bd}),\text{   and   }
\delta_q=\sum_i^\bq(1-\bI_i^{\bi})(1-\bI_i^{\bee})(1-\bI_i^{\bd})\mu_i+\sum_i\bI_i^{\bd}\zeta_i+\lambda+\gamma\label{eq_Cq}
\end{align}
\normalsize}
\end{footnotesize}
%and
%\begin{footnotesize}
%	{\footnotesize
%\[
%\delta_q=\sum_i^\bq(1-\bI_i^{\bi})(1-\bI_i^{\bee})(1-\bI_i^{\bd})\mu_i+\sum_i\bI_i^{\bd}\zeta_i+\lambda+\gamma
%\]
%\normalsize}
%\end{footnotesize}
Note that in the case where all queues are full, the outcome of the inner maximization is empty. In this case, the second term in outer maximization is selected. The maximization over deployment decision which is denoted by $\bb$ is made both in rejection and task scheduling cases. Hence, the decision to reject a task, but to start deployment of a previously idle queue is allowed. See that the product $\Pi_j\bI_j^{\bi}$ is equal to $1$ only in the case at least one queue is non-idle. Otherwise, it is equal to $0$ and no actual VM deployment happens.
In the ideal case of instantaneous VM deployment, that is, when $\zeta_i=0,\;\forall i$, we substitute $\bI_i^{\bd}=0$. The state of being deployed then does not effectively exists, hence write
\begin{footnotesize}
	{\footnotesize
\begin{align}\label{eq:b1}
& V_q=\big[\sum_i\bI_i^{\bn}\mu_iV_{q-e_i}+\sum_i\bI_i^{\bff}\mu_iV_{q-e_i}+  
\;\;\;\sum_i\bI_i^{\bo}\mu_i\max\{V_{q-e_i},V_{q-2e_i}-\psi\}+\\
&\;\;\;\lambda\max\big\{\max_{i,\bI_i^{\bff}=0}\{V_{q+e_i+\bI_i^{\bi}e_i}-\beta\bI_i^{\bi}+r\},V_{q}-f\big\}+C(q)\big]\delta_q, \no
\end{align}
\normalsize}
\end{footnotesize}
%where the cost function $C(q)$ is calculated by
%\[
%C(q)=\sum_i^\bn h_i+\kappa*(1-\bI_i^{\bi})(1-\bI_i^{\bd}),
%\]
Observe that the action space is effectively restricted, such that only one queue at each time is allowed to be deployed, at arrival opportunities. The corresponding newly arrived task will be scheduled at the newly deployed queue. While it restricts in some sense the action space, this setting is practically reasonable. 
\iflong
The derivation of equation~\eqref{eq:b1} can be found in Appendix~\ref{app:proof}.
\else
The derivation of the equation above is in the on-line version of this paper.
\fi
The importance of simplistic version of the system described by equation~\eqref{eq:b1} is mostly explorational and it is analyzed in Section~\ref{sec:Ores}. In~\ref{sec:zeta} the impact of the parameter $\zeta$ is analyzed. The complete and more realistic system version described by equation~\eqref{eq:b} is additionally validated by AWS-based implementation in section~\ref{sec:AWS}.
Bellman equation, such as~\eqref{eq:b1}, belongs to the well known class of equations which are solved by the value function $V$ which constitutes a fixed point of an operator which corresponds to the equation. In another words, one defines operator $\calT$ acting over $V$ in~\eqref{eq:b1} as follows
\begin{footnotesize}
	{\footnotesize
\begin{align}
& \calT V_q=\big[\sum_i\bI_i^{\bn}\mu_iV_{q-e_i}+\sum_i\bI_i^{\bff}\mu_iV_{q-e_i}+\sum_i\bI_i^{\bo}\mu_i\max\{V_{q-e_i},V_{q-2e_i}-\psi\}+\nonumber\\
&\;\;\;\lambda\max\big\{\max_{i,\bI_i^{\bff}=0}\{V_{q+e_i+\bI_i^{\bi}e_i}-\beta*\bI_i^{\bi}+r\},V_{q}-f\big\}+C(q)\big]\delta_q, \label{eq:b2}
\end{align}
\normalsize}
\end{footnotesize}
Since $V$ is a fixed point the display above writes $V=\calT V$. The detailed theory behind~\eqref{eq:b2} can be found in, e.g.,~\cite{bertsekas1995dynamic} and is not elaborated in this paper. We merely present the value iteration algorithm which is numerically applied in order to calculate $V$. Denote by $\calT_{\ba}$ application of the operator associated with some action $\ba\in\A$, that is, each $\calT_\ba$ refers to correspondent triplet of actions, namely, to $\ba=\{\bu,\bb,\bd\}$. Then, the value iteration algorithm is merely given in Algorithm~\ref{algo:valit}. Observe that line $7$ amounts to applying $\calT$. That is, $\max_\ba\calT_\ba=\calT$.

{	\incmargin{1em} 
	\restylealgo{boxed} \linesnumbered
	\begin{algorithm}[h]
		\begin{footnotesize}
			{\footnotesize
		\SetLine
		\label{algo:valit}
		\SetLine
		\SetKwInOut{Input}{input}
		\SetKwInOut{Output}{output}
	
		\Input{ Initial guess of $V_n$, $n\in\{1,\cdots,\rvert\Q\rvert\}$ } 
		\Output{ \begin{enumerate}
				\item Value function $V$ - fixed point solution
				\item Optimal policy $\pi^*$
			\end{enumerate} }
	
		$i \leftarrow 0$ \\
		\While {not converged}{
		
			\For {$n=1 \ldots \rvert\Q\rvert$}{
			$\triangleright$	{$\quad V_n^{\ba,(i+1)}  \leftarrow \calT_\ba V_n^{(i)}$}  \\
		
			}
		
			\For {$n=1 \ldots\rvert\Q\rvert$}{
	
				$V_n^{(i+1)} \leftarrow \max_\ba(V_n^{\ba,(i+1)})$ \\
			}
			
			$i \leftarrow i+1$ \\
			
		} 
	
		{\bf return } $(V)$
				\vspace*{0pt}
		\caption{ Value iteration for optimal VM scaling and scheduling.} 
		\normalsize}
	\end{footnotesize}
	\end{algorithm}}

Note that the size of of the state space exponentially grows with the maximal number of VMs. For example, for maximum of $5$ VMs, each one is allowed to accommodate up to $5$ tasks, we have $5$ queues with $8$ possible states each.
Hence, in this case, 
\begin{footnotesize}
	{\footnotesize
\begin{equation}\label{eq5-8}
\mq=\{-2,-1,0,\cdots,5\}, \;\; \rvert\Q\rvert=8^5=32768 \text{ \footnotesize{states}}.
\end{equation}
\normalsize}
\end{footnotesize}
Therefore, algorithm~\ref{algo:valit}, although is written for simplicity in a scalar form, was carefully treated vector-wise. See also the implementation comments in the following sections.
%We use this example system throughout the paper.
%\input{ORL}
%\input{OtimeVMnotZero}
%%%%%%%%%%%%%%%%%%%%%%%%%%%%%%%%%%%%%%%%%%%%%%%%%%%%%%%%
\section{Numerical Study of the optimal scaling and scheduling policy}\label{sec:Ores}
%%%%%%%%%%%%%%%%%%%%%%%%%%%%%%%%%%%%%%%%%%%%%%%%%%%%%%%%
As can be seen in Equation~\eqref{eq:b}, the value function $V$ is affected by many parameters hence it is hard to grasp the effects of each one on the value function, all the more so to grasp the inter-relation between the various sets of parameters and variables and their mutual effect on the value function. In this section, we provide a comprehensive study of the value functions and the corresponding policies through a comprehensive set of MATLAB simulations. The results presented herein not only explore the behavior of the value function but more importantly provide profound insight into the impact of each parameter on the value function, and consequentially on the optimal policy that should be taken. In Section~\ref{sec:AWS} we provide further insight into the mechanism suggested in this paper by partly implementing the suggested mechanism on AWS.

Since a small number of VMs (which induces a traceable state space) is sufficient to get the essence of the value function and the corresponding policies, and in order to avoid getting lost in an enormous state space, we explored a system with up to five VMs (i.e., queues). Furthermore, we simplified the model and assumed that all VMs are identical. Obviously, this limitation can be easily removed (i.e., allowing different VMs), however this option like other generalization options adds another dimension to the state-space which is already multi-dimensional, hence makes it less traceable and dilutes the insight among too many parameters and options. We start by isolating the deployment time; we assume zero deployment time ($\zeta_i=0$), and study its influence later.

%%%%%%%%%%%%%%%%%%%%%%%%%%%%%%%%%%%%%%%%%%
\subsection{Value function for negligible deployment time ($\zeta_i=0$)}
%%%%%%%%%%%%%%%%%%%%%%%%%%%%%%%%%%%%%%%%%%
In this subsection we study the value function ($V_n$) of each state in the states-space for the case $\zeta=0$. The value function attained via the Bellman equation~\eqref{eq:b1}, which is utilized to find the optimal policy, is solved iteratively via Algorithm~\ref{algo:valit}. In the sequel we will illustrate the inter-relation between the various states as obtained by the Bellman equation~\eqref{eq:b1}. Note that state $-1$ for each VM, which represents the deployment phase of the VM is omitted, as $\zeta=0$ implies zero deployment time. Figure~\ref{fig4} depicts the value function for $B=6$ (i.e., the maximal number of tasks which can be admitted by each VM is six). The horizontal axis stands for the states enumeration, where the states have been enumerated lexicographically according to the number of tasks in each VM, e.g., state $0,1,1,2,-2$ stands for zero, one, one and two tasks which are allocated in VM-1, VM-2, VM-3 and VM-4, respectively, and VM-5 is inactive.

\begin{figure}
	\begin{center}
		\begin{align*}
		{\includegraphics[width=13em]{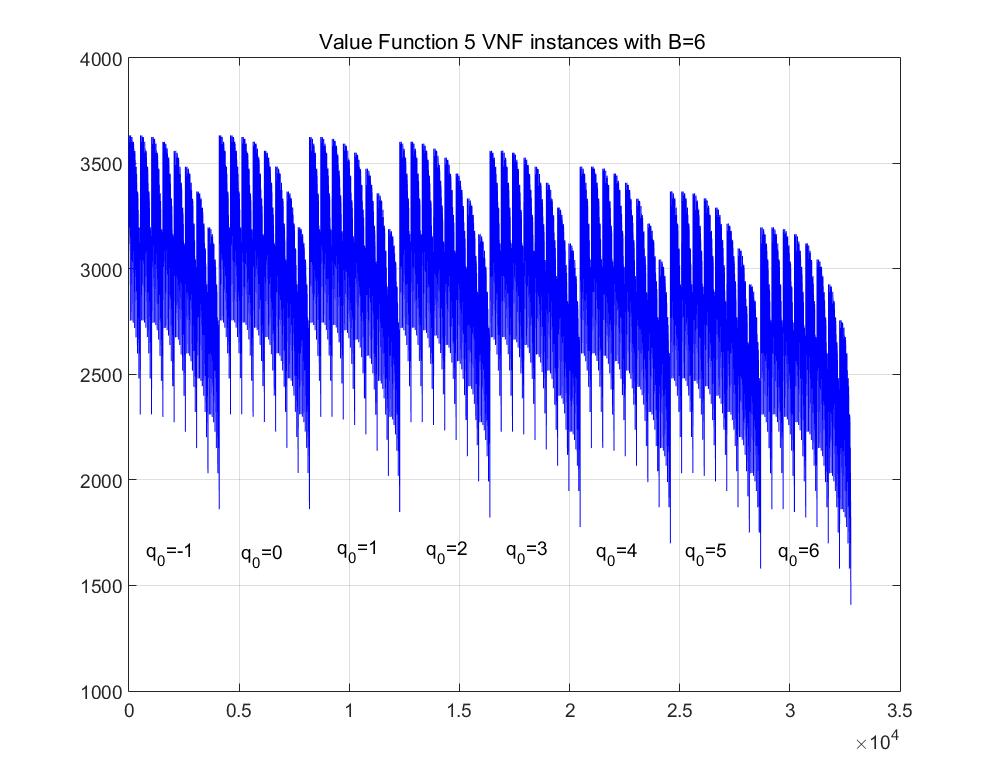}}
		{\includegraphics[width=13em]{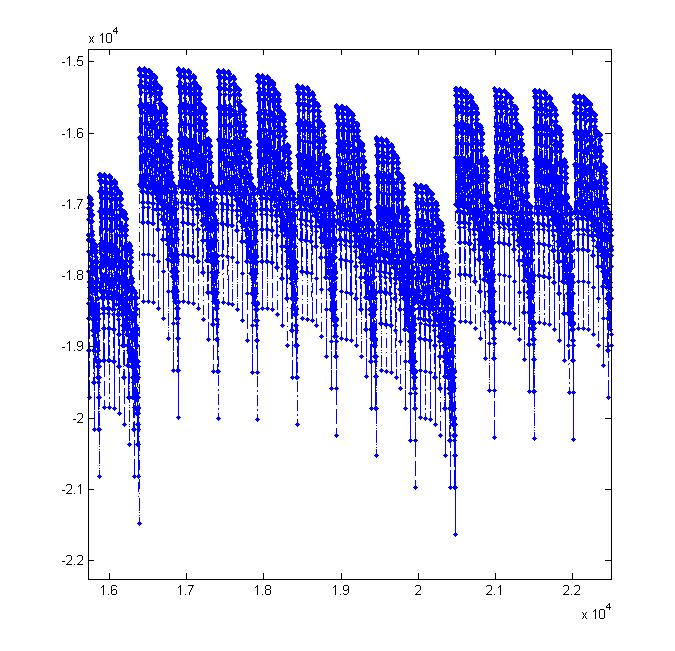}}
		{\includegraphics[width=13em]{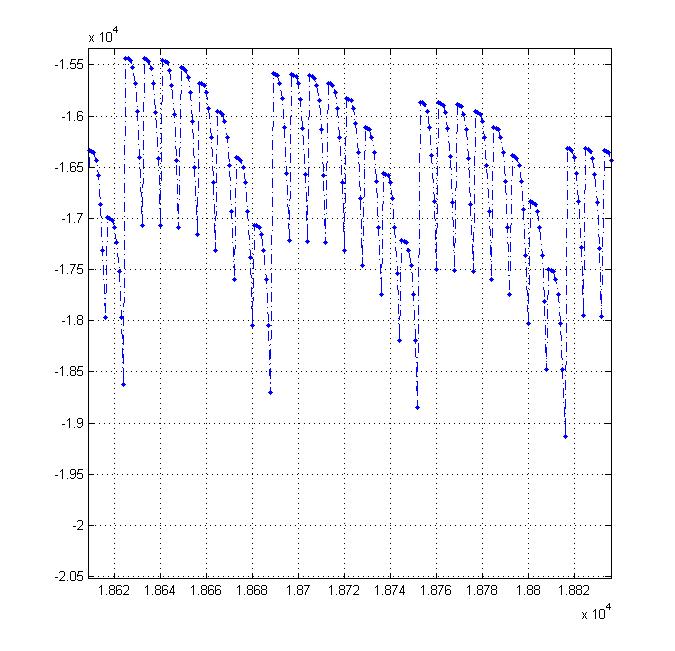}}
		\end{align*}
		\iflong\else	\vspace*{-20pt}\fi
		\caption{\sl\small
			Value function: left- all states,center -closeup over one of the regions, right - closeup inside one of the regions. Each dot corresponds to a particular $V_n$. The horizontal axis stands for the states enumeration. }\label{fig4}
	\end{center}
\end{figure}

Figure~\ref{fig4}(left) clearly depicts $8$ regions (cockscomb shape). Each region corresponds to the group of states where the number of tasks deployed on the first VM ($q_0$) was fixed. Specifically, the leftmost region corresponds to the states where $q_0$ was idle, the second from left region corresponds to the states where $q_0$ was active yet empty, the third from left region corresponds to the state where one task is active on $q_0$, etc. The regions’ shape behavior demonstrated on Figure~\ref{fig4}(left) is hierarchical, such that each region (cockscomb) comprises $8$ sub regions, each with the same cockscomb shape which also comprises an additional $8$ sub regions, etc. Each layer in the hierarchy corresponds to an additional VM with fixed number of deployed tasks (e.g., the second hierarchy corresponds to both $q_0$ and $q_1$ with a fixed number of running tasks). Figure~\ref{fig4}(midle) depicts a zoom into one of the regions depicted by Figure~\ref{fig4}(left), i.e., second hierarchy in which both $q_0$ and $q_1$ are fixed. Figure~\ref{fig4}(right) depicts the lowest hierarchy in which the number of tasks on all the VMs besides VM 5 are fixed.

Interestingly, all three figures show a decaying value function on each of the regions, which means that the more tasks are deployed, the lower the $V$. However, recall that $V$ is attained when the task is admitted, hence after its acceptance each such task’s value function is reduced by two means. First, the direct cost committed to maintain the task, and second the indirect cost due to the fact that not only the admitted task can affect the performance (cost) of the other admitted tasks, but it also occupies one of the available resources, which can potentially result in future rejection (un-admitted task), which means revenue loss. Surprisingly, also the states where no VMs are deployed ($q=-2$) attain high $V$, i.e., since the deployment costs are charged upon deployment, one could have expected that $V$ of states in which $q=0$ should be higher than those with $q=-2$. However, note the tradeoff between the deployment cost (which is already charged in the case of $q=0$) and the holding cost continuously charged for maintaining the VM. Further note, that even when the deployment delay is negligible, the strategy with respect to freeing idle VMs depends on the relation between the costs. On the one hand there is no point in baring the holding time costs keeping alive idle VMs for future use, as they can be deployed on demand whenever needed without any delay, yet on the other hand releasing a VM and re-deploying it will result in additional termination and deployment costs. Next we further explore these inter-relation costs.

%%%%%%%%%%%%%%%%%%%%%
\textbf{Keep-alive cost ($\kappa$)}
%%%%%%%%%%%%%%%%%%%%%%
We kept exploring instantaneous deployment times ($\zeta_i=0$), and examined the effect of the keep-alive cost (recall that keep-alive cost is charged per unit of time once a VM is deployed regardless of its occupancy). Besides $\zeta_i=0$ (no deployment delay) we set the task’s rejecting fine ($f$) to $10$, $\lambda=4$ and $\mu_i=1,\forall i$. We also limited the maximal number of tasks each VM can handle to 4 (Buffer size was $4$ tasks).

\begin{figure}
	\begin{center}
		\begin{align*}
		{\includegraphics[width=20em]{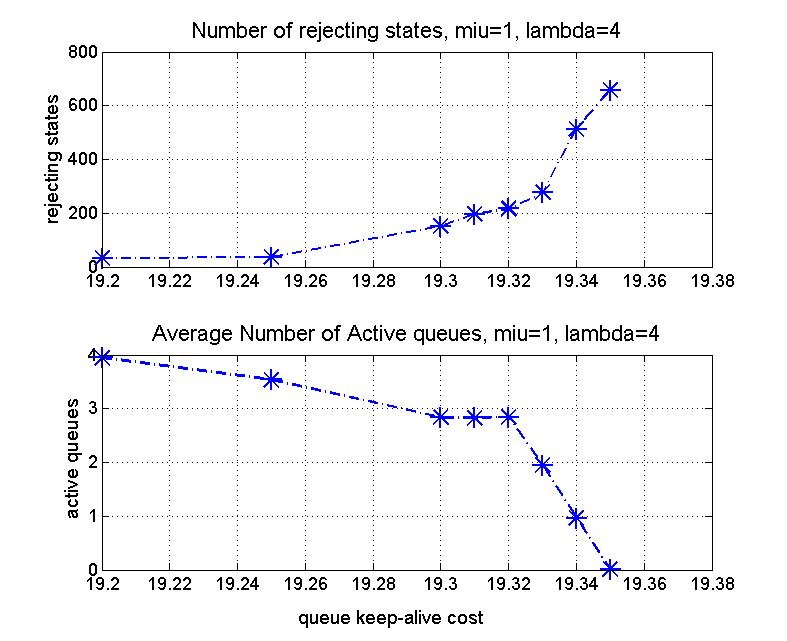}}
		{\includegraphics[width=20em]{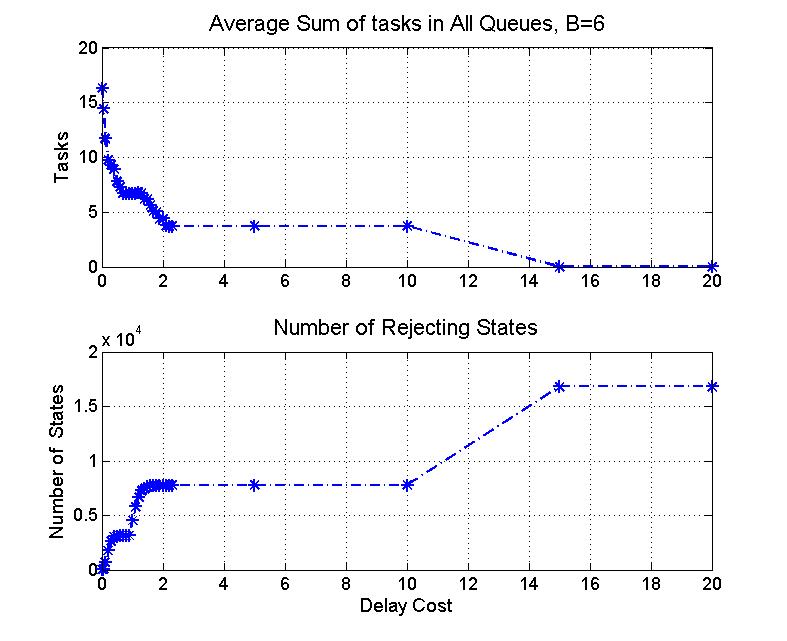}}
		\end{align*}
		\caption{\sl\small
			Left - Impact of the allocated VM cost. The upper graph shows the number of rejecting states. The lower graph shows how the number of active queues is being reduced with the cost of having an allocated VM.
			\newline Right - Impact of the delay cost.
		}\label{fig2}
	\end{center}
\end{figure}

Figure~\ref{fig2} (left) depicts the impact of the keep-alive queue cost on the task rejection states (upper left) and on the number of active VMs, lower-left. Note the tradeoff between paying the cost for maintaining VMs and the price of rejecting tasks. On the one hand due to the high task arrival rate there is an incentive to keep many VMs deployed in order to admit incoming tasks and collect the respective admission rewards. On the other hand, after admitting a task, the VM assigned to it should be kept alive at least until the task was terminated. Obviously, the lower the keep-alive cost, the more conceivable to find more deployed VMs and vice versa the higher the cost the less deployed VMs can be found.  There is a threshold for which it is more affordable to pay the fine of rejecting incoming tasks rather than maintaining a VM (i.e., a queue), which indeed can be seen in both Figure~\ref{fig2} panels on the left. Above this threshold ($\kappa \sim 19.35$) all states lead to reject policy and there are no active VMs. Note that deploying a VM due to temporal congestion can result in paying the price of maintaining an excessive number of unutilized VMs for a long time period, especially when migration of tasks between VMs is not supported. Accordingly, the decision maker should balance between the number of active VMs and the tasks’ arrival rate which is reflected by the load. Specifically, when the keep alive cost is high, the decision maker should try to maintain less yet congested VMs, at the price of rejecting a task once in a while. Indeed, as can be seen in the figure, the decline from keeping all VMs alive and keeping no VMs alive is not strict and there is a keep-alive cost range at which the number of deployed VMs gradually declined from all to no deployed VMs. Note that the decline from 3 deployed VMs to 2 and later to 1 or zero was much sharper than the slope between 4 to 3 deployed VMs.

%%%%%%%%%%%%%%
\textbf{Delay cost ($h$)}
Next we examined the cost delaying tasks (Figures~\ref{fig2} (right)). Note the tradeoff, on the one hand in order to keep delay low, one needs to preserve many active VMs and distribute the load among them. On the other hand, preserving many active VMs results in high keep-alive cost. The arrival intensity was $4.75$ and $\mu_i=1,\forall i$. The buffer size of each queue was $6$. We set the keep-alive cost to $1$. Rejecting fine was set to $10$. The delay constants we utilized were $\eta_1=1,\eta_2=1.8,\eta_3=2.5,\eta_4=3.5,\eta_5=4.5,\eta_6=5.5$, for having 1, 2, 3, 4, 5 and 6 operational tasks on a VM, respectively.
Figure~\ref{fig2} panels on the right depict the average number of managed tasks (upper) and the number of rejected states (lower). Figures~\ref{fig2}.  Observe that the number of rejecting states approached 100\% of all states at highest values of $h$.

Since throughout this simulation setup, the keep-alive cost was sufficiently low compared to the admission gain (i.e., no significant keep-alive to delay tradeoff), all the VMs were active. As expected, the scheduler balanced between the loads of the operating VMs (as anticipated by the cost model in Equation~\eqref{1b}). Expectedly, as long as the delaying costs were low, most of the tasks were admitted, and very few states were rejecting states. When the delay cost increased, keeping several tasks on a VM degraded the performance (the admission gain of a single task was lower than the extra delay cost incurred by all tasks on the designated VM). When the delay cost was sufficiently high (around 1.8) having more than one task per VM was costly, hence we could see only 4 operational tasks, one per VM. Note that the interval of keeping a single task per VM was quite large, since the delay cost when a single task was operational on a VM ($\eta_1$) was low. When the delay cost was high enough (around 15) the admission gain could not cover the task maintenance costs and all tasks were declined.

We inspected many other parameters analyzing the tradeoffs between different costs, searching for threshold policy, and trying to understand the interdependencies between parameters. For example, a threshold policy can be observed with respect to the deployment and termination of VMs as a function of several parameters and their interdependencies. For example, if $\beta$ and/or $\psi$ are high compared to keep-alive cost, the optimal policy acts to leave all queues active, even if empty. Clearly, the set of system parameters and most importantly their proportional relation will determine the optimal policy, e.g., will determine whether to admit or reject tasks, whether load balance between VMs to reduce the delay or to shift loads to less VMs in order to release a VM, etc. However, due to space limitations we only provide a sample of our results and only a glimpse at few observations, to exemplify the scheme usage.

%%%%%%%%%%%%%%%%%%%%%%%%%%%%%%%%%%%%%%%%%%%%%%%%
\subsection{Impact of VM time deployment - the $\zeta_i>0$ case} \label{sec:zeta}
%%%%%%%%%%%%%%%%%%%%%%%%%%%%%%%%%%%%%%%%%%%%%%%%
After ignoring the deployment delay, assuming that VMs can be deployed and terminated  instantaneously, in this subsection we explored the effect of deployment delays.
Obviously, the time which takes to activate and inactivate VMs can have a major effect on the policy.%\mmme{we only explored in depth deployment A leave as is}

In practice, the length of this period depends on several aspects including the type of the machine. For example, the applications which we dispatched for the execution on AWS in our implementation (Section~\ref{sec:AWS}) were deployed and run on VMs of type "tx4.large". These VMs normally took between half a minute and two minutes to deploy and boot the image. We repeated the numerical evaluation utilizing the same parameters as before, varying the deployment delay ($\zeta$). Note that $\zeta$, in our notation, stands for the rate (i.e., the reciprocal of the time).

\begin{figure}[h]%{R}{0.6\textwidth}
	\begin{center}
		\begin{align*}
		{\includegraphics[width=20em]{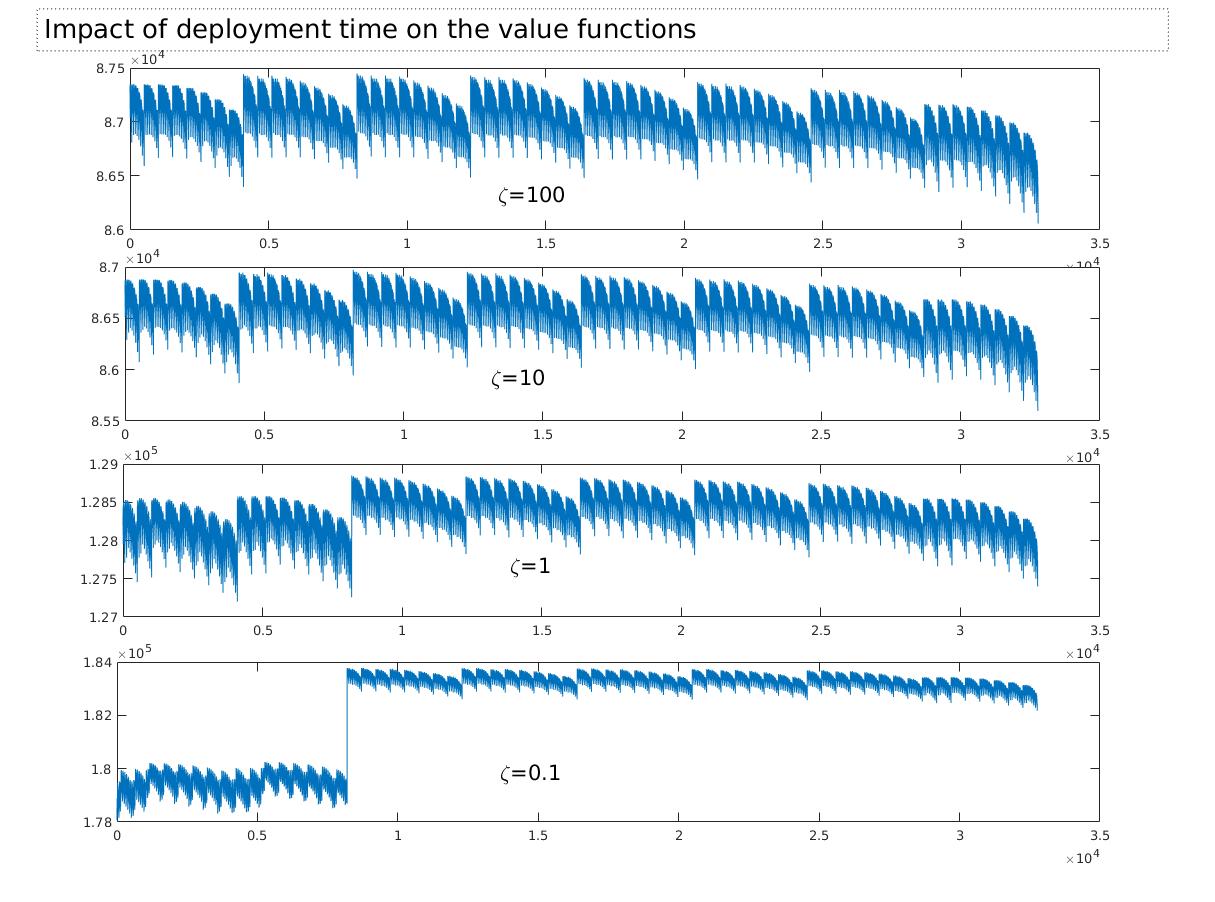}}
		{\includegraphics[width=20em]{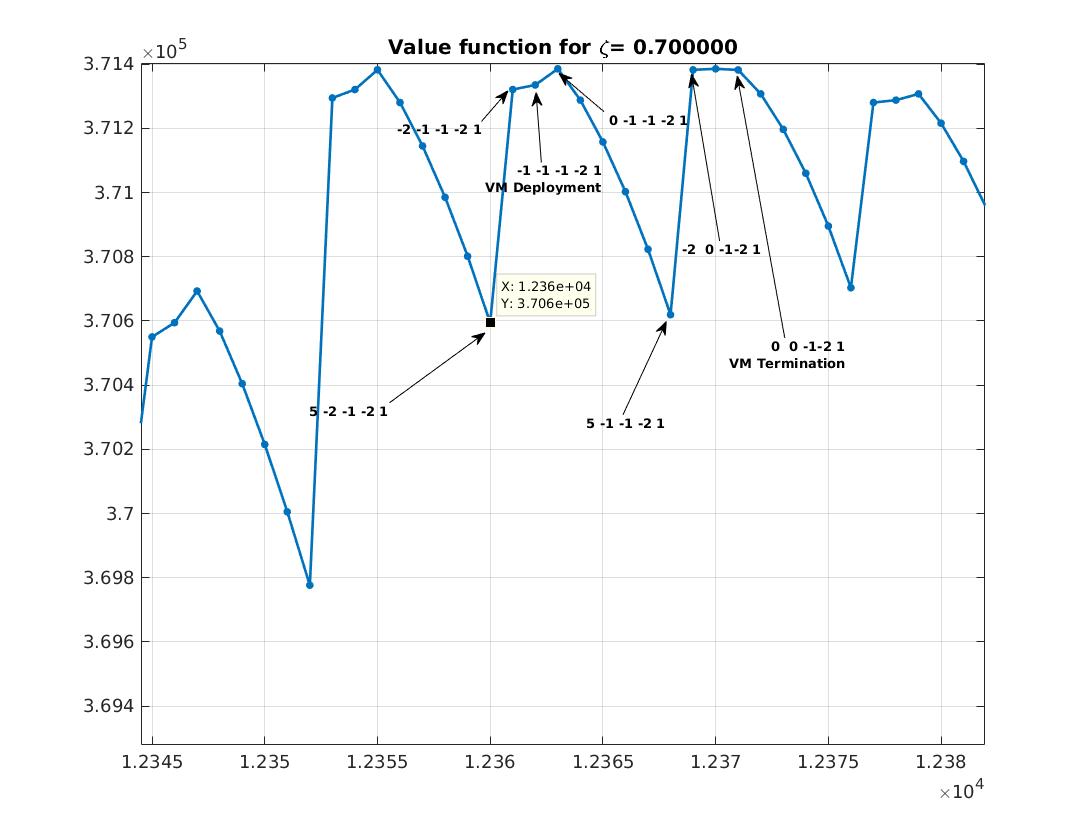}}
		\end{align*}
		\caption{\sl\small
			Left: Value function - impact of $\zeta$ is explored.
			Right: Value function - detailed view on the impact of $\zeta$.
		Vertical axis shows the value function and the horizontal axis enumerates the states in all cases. }\label{figz2}
	\end{center}
\end{figure}

The impact of the deployment delay on the value function is depicted in Figure~\ref{figz2}. It is apparent (Figures~\ref{figz2}(left)) that the “crest” type behavior repeats itself due to the same rationale as before, with the exception of additional perceptible deployment states ($q=-1$). Obviously, when the deployment rate is high (meaning negligible deployment delay) these states are almost unobserved. Furthermore, when the deployment rate is high the results resemble the results attained for $\zeta=0$. However, when the deployment rate is low (i.e., long implementation delay) the value function of the two leftmost graphs which correspond to the cases $q_5=-2$ and $q_5=-1$, i.e., idle and deploying, are significantly lower than those of the other graphs. As in the case of $\zeta=0$ the reward for admitting a task was attained on admission. However, in
contrast to $\zeta=0$ for which the deployment time is negligible, now the deployment time affects
$V$, specifically the higher the deployment time, the higher the effect on the value function.  Note that the decreased value function is due to the delay imposed on the waiting and incoming tasks and the increased future tasks’ rejection probability.

In order to better understand the impact of the deployment delay, Figure~\ref{figz2}(right) concentrates on a specific delay, namely $\zeta=0.7$ (relatively high deployment delay). Let us focus on states in which the loads of all VMs are fixed, varying the number of tasks occupying VM-1. Specifically, we examine three adjacent states: $q^1=\{-2,-1-1,-2,1\},q^2=\{-1,-1-1,-2,1\},q^3=\{0,-1-1,-2,1\}$, and denote their corresponding value functions as $V(1)$, $V(2)$ and $V(3)$, respectively. Under the specific parameters (and specifically deployment and arrival rates and , especially their ratio),  $V$ of having a deployed VM and despite the keep-alive cost, is higher than $V$ where VM is inactive. The marginal difference between states that have VM under deployment ($q^2$ ) and inactive ($q^1$) is smaller compared to the difference between $q^3$ (idle VM) and $q^2$ (under deployment). That is, $V(3)-V(2)>V(2)-V(1)$, indicating that the advantage of having a pending request for VM deployment was less valuable than that of having an already deployed empty queue, once $\zeta$ is small. Note that $V(2)-V(1)$ captures the value of taking the decision of VM deployment. As before, the more loaded the VMs, the lower the value; for example the value of state $\{1,-1-1,-2,1\}$, in which a single task occupies VM-1, is lower than that of state $\{0,-1,-1,-2,1\}$ in which VM-1 is deployed but idle. More radical changes can be seen between $V$ of state $\{5,-1,-1,-2,1\}$, in which VM-1 is fully loaded, and that of state $\{0,-2,-1,-2,1\}$. Recall that the reward  for all the extra admitted tasks has already been obtained on admission. For comparison, observe the three adjacent states $q^4=\{-2,0-1,-2,1\},q^5=\{-1,0-1,-2,1\},q^6=\{0,0-1,-2,1\}$, which are different from the previous triplet by that the second VM state changed from $-1$ to $0$, i.e., having VM-2 deployed idle and disposed to accept new tasks. Obviously, the need for a ready unloaded VM now is less acute than before, i.e., the marginal contribution of an additional ready VM is less acute than before. Indeed, as is apparent in the figure, there is no real value difference between the three states. Further, their value function is even slightly lower than $V$ of state $q^3$, i.e., the implementation of a VM when there is already an idle VM ready to accept tasks slightly degrades the value function due to the expected keep-alive costs. Following the decision mechanism one can see that indeed VM $q^6_5$ is marked for termination if its only packet is served and it becomes empty before any other event. (recall that decisions are taken only as a consequence of an event hence the VM termination must be triggered by a service completion event and cannot be done afterwords then the queue was already empty).

%%%%%%%%%%%%%%%%%%%%%%%%%%%%%%%%%%%%%%
\subsection{Threshold-type structure of the optimal policy}
%%%%%%%%%%%%%%%%%%%%%%%%%%%%%%%%%%%%%%
In this subsection we give some insight into finding optimal policies (policies that maximize the expected utility). In particular, we concentrate on the structure of the optimal policies trying to define thresholds such that below such a threshold the system takes one action while above it, it takes a different action. For example, the decision maker keeps admitting tasks to a VM only below a number of tasks occupying this VM and above this threshold it will either assign new incoming tasks to a different VM (possibly new one) or reject them. The motivation for identifying threshold policies stems from the fact that on many systems, and in particular queueing systems, threshold policies are optimal or nearly optimal. We mainly concentrate on actions which result in the deployment or termination of a VM and on the load balancing policy on which VM to place an admitted task.

In order to define threshold policy, we first define state domination. Consider two states $a$ and $b$, with queue vectors denoted by $q^a$ and $q^b$.

\begin{footnotesize}
	{\small
		\begin{definition}[State domination]\label{def4}\hfill\newline
			We define that state $a$ \textit{dominates} state $b$ if and only if states $a$ and $b$ have the same number of idle,  deployed and under deployment queues, and $q_i^a\geq q_i^b,\;\;\;\forall q_i^a>0,\; i\in\{1,\ldots,\bn\}$. We denote such domination by $q^a\succeq q^b$.
		\end{definition}	
		\normalsize}
\end{footnotesize}

The following defines thresholds in build (VM deployment), scheduling, and destroy (VM termination):
\begin{footnotesize}
	{\small
		\begin{definition}[Threshold policies]\label{def2}\hfill\newline
			\vspace*{-20pt}
			\begin{itemize}
				\item Optimal threshold policy $\pi^\bb$ exists if a deployment of previously inactive VM at state $q^a$ means deployment is also optimal at all states $q^b$ such that $q^b\succeq q^a$
				\item Optimal threshold policy $\pi^\bd$ exists if an optimal termination queue at state $q^a$ means termination is also optimal at all states $q^b$ such that $q^a\succeq q^b$.
				\item Optimal threshold policy $\pi^\bu$ exists if the optimal load balancing policy in state $q^a$ is the same as the optimal scheduling policy for all states $q^b$ such that $q^b\succeq q^a$ and $q^b_i=q^a_i$.
			\end{itemize}
		\end{definition}
		\normalsize}
\end{footnotesize}
We identify the existence of threshold policies both by observation and by using the following analytical result which states the monotonicity of the value function:
\begin{lemma}[Value function domination]\label{lem:dom}
	For any $q^a\succeq q^b$ it holds $V(q^a)\leq V(q^b)$.
\end{lemma}
The intuition behind this Lemma is quite straightforward; as previously explained, the reward for admitting a task was attained on admission, hence after admitting a task, it is only a ”burden” on the value function, hence the more tasks are present in a VM the lower the value. This monotonicity is clearly depicted in Figure~\ref{figz2}(right). For example, the four states to the left of state $\{5,-1,-1,-2,1\}$ which correspond to states $\{4,-1,-1,-2,1\}$,$\{3,-1,-1,-2,1\}$,$\{2,-1,-1,-2,1\}$ and $\{1,-1,-1,-2,1\}$ have a decreasing number of tasks on VM-1 and the same number of tasks on all other active VMs, their value gradually increased accordingly. The formal proof to Lemma~\ref{lem:dom} appears in Appendix~\ref{sec:lemdom}.

Lemma~\ref{lem:dom} analytically states a structural property of the value function, which suggests the existence of threshold policy. Next we illustrate these threshold policies on our numerical results. Note that these threshold policies coincide with the general intuition previously explained. In particular, since the action $\bb$ is motivated by the intention to reduce future costs due to the loads on the active VMs which include delay costs and potential fines, the same action is expected to be even more essential in the dominating states in which the VMs are even more loaded.

\begin{table}[t]
	\begin{minipage}{.55\linewidth}
		\vspace{0.1in}
		\center
		\begin{footnotesize}
			{\footnotesize
				\begin{tabular}{l|l|l|l|l|l|l|l}
					
					\hline
					state &	$q_1$ & $q_2$ & $q_3$ & $q_4$ & $q_5$ & $\bb$ & $\bu$ \\ %  & comment \\
					
					\hline
					$q^1$ &	0 & -2 & 5 & 3 & 4 &0 &1 \\
					$q^2$ &	1 & -2 & 5 & 3 & 4 &1 &1\\
					$q^3$ &	2 & -2 & 5 & 3 & 4 &1 &1\\
					$q^4$ &	3 & -2 & 5 & 3 & 4 &1 &1\\
					$q^5$ &	4 & -2 & 5 & 3 & 4 &1 &4\\
					$q^6$ &	5 & -2 & 5 & 3 & 4 &1 &4\\
					\hline
				\end{tabular}
				\normalsize}
		\end{footnotesize}
		\vspace{0.1in}
		\caption{\scriptsize{Building and scheduling policy. System parameters were $\lambda=2,\mu=1.4,\zeta=0.1$, with cost set given by
				$\{r,f,\beta,\phi,h,\kappa\}=\{100,5,10,0,1,120\}$.}} \vspace{-0.3in}  \label{tab:1}
		%\end{table}
	\end{minipage}
	\begin{minipage}{.5\linewidth}
		%\begin{table}[t]
		\vspace{0.1in}
		\center
		\begin{footnotesize}
			{\footnotesize
				\begin{tabular}{l|l|l|l|l|l|l|l|l}
					
					\hline
					state &	$q_1$ & $q_2$ & $q_3$ & $q_4$ & $q_5$ & $\bb$ & $\bu$ &$\bd$ \\ %  & comment \\
					
					\hline
					$q^1$ &	0 & -1 & 5 & 3 & 4 &0 &1&0 \\
					$q^2$ &	1 & -1 & 5 & 3 & 4 &0 &1&0\\
					$q^3$ &	2 & -1 & 5 & 3 & 4 &1 &1&0\\
					$q^4$ &	3 & -1 & 5 & 3 & 4 &1 &1&0\\
					$q^5$ &	4 & -1 & 5 & 3 & 4 &1 &0&0\\
					\hline
					$q^6$ &	0 & 0 & 1 & 1 & 1 &0 &1&1\\
					$q^7$ &	0 & 0 & 1 & 1 & 2 &0 &1&1\\
					$q^8$ &	1 & 0 & 1 & 1 & 2 &0 &2&0\\
					$q^9$ &	1 & 1 & 1 & 1 & 2 &0 &1&0\\
					\hline
				\end{tabular}
				
				\normalsize}
		\end{footnotesize}
		\vspace{0.1in}
		\caption{\scriptsize{Building and scheduling policy. System parameters were $\lambda=4,\mu=1.4,\zeta=1.1$, with cost set given by
				$\{r,f,\beta,\phi,h,\kappa\}=\{100,60,100,0,1,40\}$,\vfill Load impact $\eta=\{1, 1.8, 2.5, 3.5, 4.5\}$.}} \vspace{-0.3in}  \label{tab:2}
	\end{minipage}
\end{table}
\begin{comment}
\iflong
The proof appears in Appendix~\ref{ap:p2}.
\else
The proof appears in the on-line version of the paper.
\fi
\end{comment}

Table~\ref{tab:1} depicts six different states, satisfying domination relation such that $q^6\succeq q^5\succeq q^4\succeq q^3\succeq q^2\succeq q^1$. The table also depicts the optimal action that should be taken based on our numerical evaluation. In particular, $\pi^\bb$ denotes the VM’s deployment action, where 1 in the table implies a deployment of a new VM, and 0 denotes no change. $\pi^\bu$ refers to the load balancing action and in particular on which VM to allocate a new admitted task.
The table clearly depicts a threshold policy for both actions. Regarding deployment (column $\bb$ in the table), state $q^1$ with idle VM-1 indicates that no action is required (no need for additional VM deployment), all other states with increasing number of tasks allocated to VM-1 necessitate a VM deployment. Note that determining the optimal thresholds, even for this seemingly simple case, is quite complicated. On the one hand, the decision maker needs to keep the number of deployed VMs small in order to minimize the keep-alive costs. On the other hand, it needs to take into account many other parameters such as the deployment delay, the expected tasks’ arrival and service rates, the delay costs associated with the number of deployed VMs, etc. As can be seen in Table~\ref{tab:1}, despite the relatively high keep-alive costs, a single task on VM-1 was sufficient to trigger a VM deployment. Regarding the Load Balancer (column $\bu$ in the table), as long as VM-1 is the least loaded active VM, new admitted tasks will be assigned to it. As soon as VM-4 becomes the least loaded VM (state $q^5$), new admitted tasks are assigned to it. Recall that due to the holding cost ($h>0$) and the homogeneity of the VMs, the  Load Balancer equally disperses between the active VMs.

Table~\ref{tab:1} depicts $\pi^\bb$ and $\pi^\bu$. See that $q^6\succeq q^5\succeq q^4\succeq q^3\succeq q^2\succeq q^1$.
Table~\ref{tab:2} depicts the value function of selected states under different setup, and also examines the VM release termination action. Recall that the decision maker needs to take into account the tradeoff between keeping more active VMs with low anticipated delays vs. a lower number of VMs, reducing the keep-alive costs yet increasing the delay costs. Recall that the delay cost is a function of the constant $h$ multiplied by factors selected from the increasing sequence $\eta$ according to the number of tasks in each queue. For example, in the setup associated with Table~\ref{tab:2}, which utilizes the sequence $\eta=\{1, 1.8, 2.5, 3.5, 4.5\}$, the total holding cost per unit of time at state $q^5$ is equal to:
\[
\sum_ih_i(t)=(4*3.5+0+5*4.5+3*2.5+4*3.5)*h
\]
Note that in this setup $\eta$ implies very high penalization for highly loaded VMs.
Further note the first five states in the table II clearly depict a deployment threshold (related to state $q^3$). Interestingly, due to the high keep-alive cost, the optimal policy for state $q^5$ is to reject an incoming task rather than assigning it to one of the VMs, even though not all of them are saturated. The last four states ($q^6-q^9$) clearly suggest threshold policy for terminating a VM (related to state $q^8$), i.e., the optimal policy for state $q^7$ is VM termination (denoted by 0 on column $\bd$); the same termination policy is also valid to states $q^8$ and $q^9$ which dominate $q^7$ ($q^9\succeq q^8\succeq q^7$).

%\input{awsSim10}

%%%%%%%%%%%%%%%%%%%%%%%%%%%%%%%%%%%%%%%%%%%%%%%%%%%%%%
\section{AWS-based policy validation}\label{sec:AWS}
%%%%%%%%%%%%%%%%%%%%%%%%%%%%%%%%%%%%%%%%%%%%%%%%%%%%%%
In the previous section we utilized MATLAB simulations to understand the effects of several representative parameters on the overall value, and to explore the inter-relation between various sets of parameters and variables and their reciprocal effect on the value function. In this section, we demonstrate the feasibility of the scheme on Amazon Web Services (AWS) \footnote{AWS provides on-demand cloud computing platforms on a paid subscription basis.} based settings. Specifically, we ran numerous experiments on AWS, extracted the required parameters and assessed the performance based on the formulation described in Section~\ref{sec:Ores}.

The common scaling policy which is also adopted by the default AWS scaling mechanism is a threshold based policy and concerns the deployment and termination of VMs, i.e. deployment and termination of VMs are according to predefined thresholds~\cite{awsAS}. These thresholds are associated with certain system-related metrics, such as average load on all deployed VMs, most loaded VM, least loaded VM, CPU utilization, etc. While being robust, effective and not less importantly simplistic, predefined threshold mechanisms significantly limit the opportunities for cost optimization. Specifically, the AWS threshold based scaling mechanism imposes rigidity which not only requires the user to determine the scaling thresholds \emph{a-priori} but more importantly provides no scaling tools which respond according to the detailed states of each VM, hence limits the range of attainable solution. Furthermore, even though the AWS threshold based scaling mechanism provides a wide range of flexibility, allowing the user to program its own thresholds \emph{a-priori}, it restricts the thresholds to rely on the system parameters available through the AWS console (which represent maximum or average load for all active VMs), highly limiting the user’s flexibility. Both reasons prevented us from deploying the complete suggested scheme on AWS and forced us to utilize a hybrid scheme which interlace the AWS platform with adjacent offline Matlab value-computations. Specifically, throughout this section we rely on qualitative validation, i.e., we extract all system states, parameters, statistics and performance values in real-time from AWS throughout each evaluation test, yet the eventual cost was computed offline.
\begin{wrapfigure}{R}{0.5\textwidth}
	%	\begin{center}
	%\begin{align*}
	\includegraphics[width=20em]{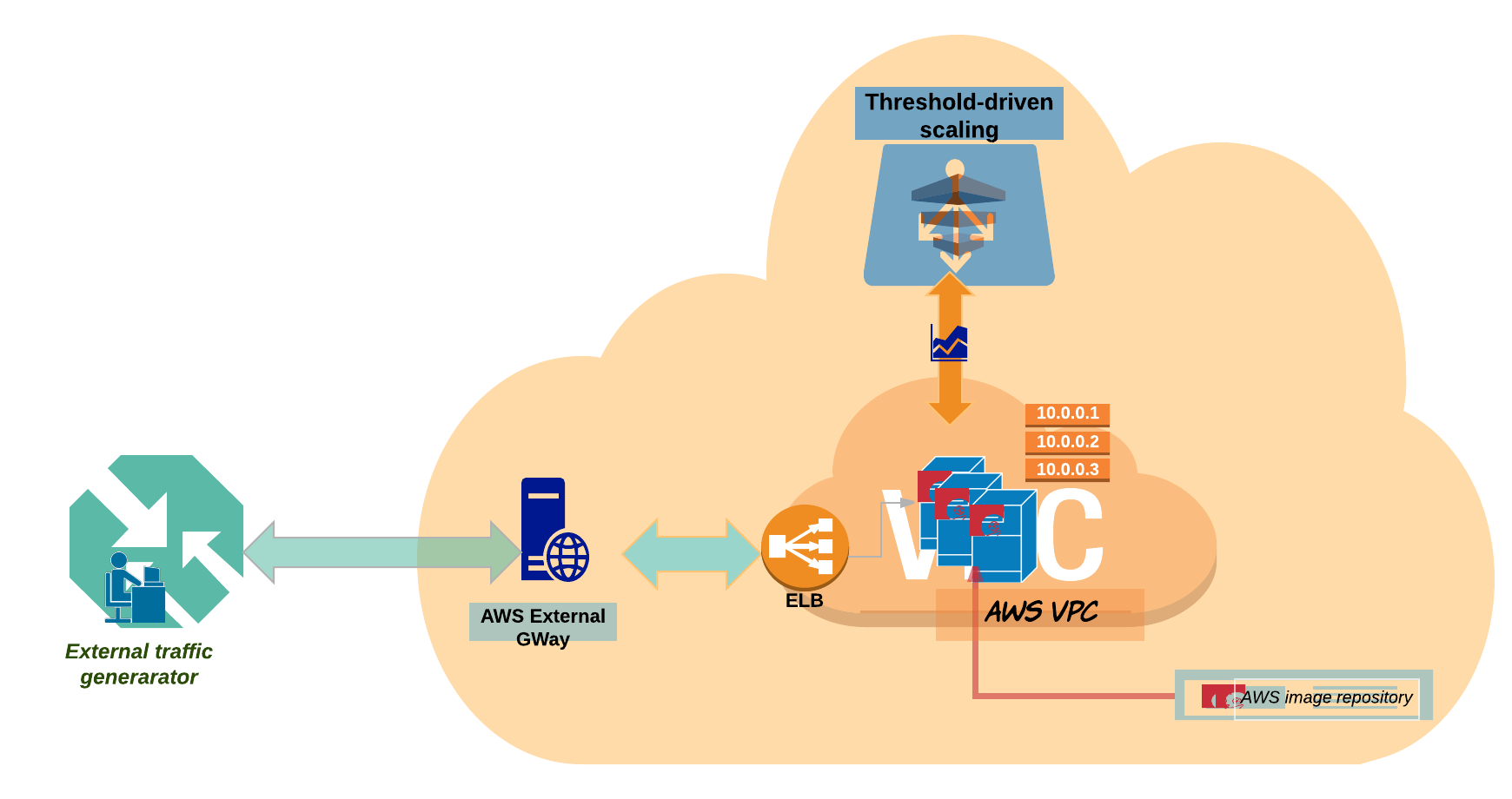}
	%\end{align*}
	\vspace*{-30pt}
	\caption{\sl\scriptsize
		AWS validation set-up scheme}\label{fig33}
	%	\end{center}
\end{wrapfigure}
The setting which we implemented on AWS is schematically demonstrated in Figure~\ref{fig33} and is summarized as follows:
\begin{itemize}
	\item \textbf{Traffic generator}.
The traffic generator resides on a local computer outside the AWS premises. Specifically, we implemented a simple JAVA HTTP client on a local computer which periodically sent HTTP task requests to the AWS gateway (GW) attached to the designated Elastic Load Balancer (ELB). In this validation scenario, we mainly focused on computational oriented tasks, rather than on NFV related tasks. In particular, each CPU in a deployed VM performed mathematical operations, e.g., matrix inversions. These tasks can fall well within a variety of jobs needed for image processing.
	\item \textbf{AWS EC2 virtual machines.}  Each VM ran an HTTP server. Once an HTTP request arrived to the server, it triggered a computational task such as matrix inversion, that loaded one of the VM CPU cores for a few seconds. We used machines of the type ” t2.xlarge” which contained 4 cores, hence were able to concurrently process up to $4$ tasks without noticeable slowdown in performance. Each task was executed by a thread running on a separate core. All VMs were configured such that in case this predefined limit number of $4$ running threads was reached, the incoming new tasks were rejected.
	\item \textbf{AWS Elastic Load Balancer (ELB)}. We utilized the AWS built-in load balancer, which equally disperses the incoming HTTP requests by the round robin (RR) method, regardless of the load level at each VM.
	\item \textbf{AWS Auto Scaling configuration.} As previously mentioned, the AWS AutoScale monitors a certain metric; once this metric crosses an upper threshold, AWS will automatically launch a new copy of our VM, and register it in the ELB, reducing the load at the currently running VMs. We utilized average CPU load as our metric for the upper threshold. The Auto Scaling configuration also specifies the maximal number of VMs and the thresholds for opening/closing a VM; it includes cloud-watch alarms which signal to deploy or to terminate a VM according to the policy. Each VM also listened on another port for "keep alive messages" from the ELB. %, to avoid these messages from adding load to the VMs (FIXME*** another   => not adding load?***) . Note that the alarms’ statistics and the logs were retrieved by a specially designated environment~\cite{boto}.\mmme{It had to be on another port to avoid TCP collisions. We can remove the load explanation, we do not need it }
\end{itemize}

%\FMO{***} -  \FMO{is this part of the setup???} \mmme{Yes, was a part of Java server implementation. The limit was 4.  }

As previously mentioned, AWS scaling mechanism relies on threshold based policy, i.e., a user can control when to scale-in or scale-out by choosing one or more system parameters (e.g., load) and by determining a threshold according to these system parameters for deployment and termination of VMs.
As was shown in previous sections by observation of numerical results, the optimal policy disclosed by the MDP is a threshold based policy (at least with respect to some of the parameters examined).
Accordingly, our objective was to understand the effect of these threshold values on the value attained, and to compare it with the optimal thresholds attained by the MDP formulation. In order to evaluate the values attained by each set of thresholds experimentally, and in order to formulate the MDP and solve the corresponding Bellman equation for determining the optimal value analytically, we need to acquire the required system parameters. For example, we need to know the effect of task loads on the performance, i.e., to acquire the statistical properties of task processing-time-distribution under different VM load levels, we need to obtain VM deployment time, etc. This data was extracted by the AWS Cloud-Watch mechanism throughout each experiment and fed back to the MDP formulation. %Obviously, some of these performance values are policy dependent, i.e., vary under different policies (e.g., \FMO ***such as…***) .
\begin{comment}
\FMO{*** Accordingly, we implemented an iterative mechanism (Figure~\ref{figCL} ***FIXME Figure 1***), in which the performance values extracted in the first iteration were fed back to the MDP procedure which in turn determined the policy for the next iteration and were fed back to the ELB as new threshold values, and so on and so forth, until there were no policy changes between iterations.} (\FMO{Frankly, I do not think that this is what we did}). In most experiments a small number of iterations was sufficient to converge to the optimal policy.  ***\FMO{check if any of this was done} \mmme{No more than 1-2 iterations was done (as we talked)}
\end{comment}
We evaluated the system under moderate traffic intensity (Arrivals were about 300 tasks per hour, while service rates at each VM were about 50 tasks per hour). The AWS VM deployment and termination thresholds were set with respect to CPU utilization (i.e., the average CPU utilization on all deployed VMs), examining the effect of various thresholds for deploying (scale-out) and discharging (scale-in) VMs, on the value function. Each pair of thresholds was examined for a long duration %(***\FMO{XXX hours?***})\mmme{2 hrs, but can we omit this?} 
to get sufficient statistics, extracting all the required system parameters. The associated values were computed offline based on the collected traces. Since the main goal of this deployment is a proof of concept, and due to budget constraints, we examined only several thresholds. The set of coupled threshold values examined are given in Table~\ref{tab:4}. For tractability we indexed the paired values from 1 to 12 (e.g., pair number 4 denotes $40\%$ CPU-utilization and $60\%$ CPU-utilization for deploying and discharging a VM, respectively). We utilized the Auto Scaling Instance Protection scheme offered by AWS (\cite{scprot}), ensuring that at least one VM is active at all times, even if its load measure is below the discharge value. It is important to note that the AWS Elastic Load Balancing Connection Draining mechanism ensures that VM instances are removed from service progressively (\cite{drain}). Specifically, after initiating a VM termination, the Load Balancer will allow existing, in-flight requests to complete, but will not send any new requests to this instance. The instance termination will be eventually finished only when no tasks are left. We charged the user for the termination period.

\begin{table}[t]
	\vspace{0.1in}
	\center
	\begin{footnotesize}
	{\footnotesize
	\begin{tabular}{l|l|l|l}
		%1     2     3     4     6     7     8    10    11    12

		\hline
		\backslashbox{Exp. \#}{Threshold}
		 &Termination	 & Deployment \\ %  & comment \\
		
		\hline
		1 &	20 & 40 \\
		2& 25 & 45&\\
		3 &	30 & 50 \\
		4 &	35 & 65 \\
	%		5&&	45 & 55 &\\
		5&	50 & 70 \\
		6&	60 & 80 \\
	7&	70 & 90 \\
	%		9&1&	75 & 95 \\
			8&	10 & 80 \\
			9 & 10& 20 \\
				10&	30 & 70 \\
		%	MATLAB &1& 49.8 & 70.1\\
		%	MATLAB &2& 30 & 56.5\\
		%	MATLAB &3& 18.2 & 30\\
		%	MATLAB &4& 31.2 & 52.1\\
		\hline
	\end{tabular}
\normalsize}
\end{footnotesize}
	\vspace{0.1in}
	\caption{\scriptsize{Scaling thresholds. AWS thresholds were tested for all 4 experiment sets. The first threshold is for VM deployment while the second threshold stands for termination, measured in average $\%$ of CPU occupancy}} \vspace{-0.3in}  \label{tab:4}
\end{table}

%%%%%%%%%%%%%%%%%%%%%%%%%%%%%%%%%%%%%%%%%%%%%%%%%%%%%%%%%%5

We examine the value of the different thresholds under four different sets of costs. The parameters used in each set are summarized in Table~\ref{tab:Costs} (the exact costs of all the four sets are given in Figure~\ref{fig34} caption). The results are summarized in Figure~\ref{fig34}.

\begin{table}[t]
	\vspace{0.1in}
	\center
	\begin{footnotesize}
	{\footnotesize
\begin{tabular}{|l||*{5}{c|}}\hline
\backslashbox{Set \#}{Property}
&Reward ($r$)& Penalty ($f$) & Deployment cost ($\beta$) &keep-alive cost ($\kappa$) \\\hline\hline
Set I &Average&Low &Moderate &High  \\\hline
Set II &Average& High &Moderate &Moderate\\\hline
Set III &Average& High &Low &Low\\\hline
Set IV &Average& High &Low &Moderate\\\hline
\end{tabular}
\normalsize}
\end{footnotesize}
	\vspace{0.1in}
	\caption{\scriptsize{Qualitative summary of the four tested scenarios}} \label{tab:Costs}
\end{table}

 \begin{figure}[h]
	\begin{center}
		\begin{align*}
	&	{\includegraphics[width=15em]{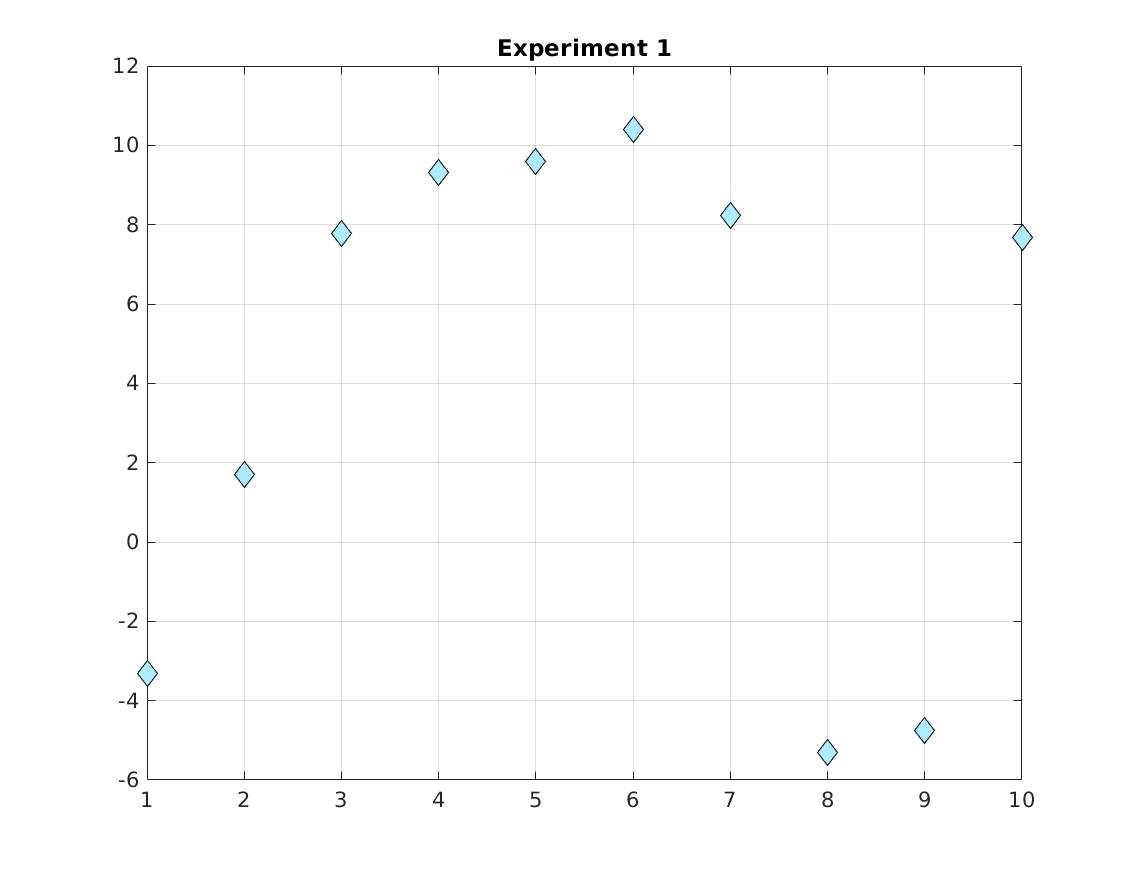}}\; 	{\includegraphics[width=15em]{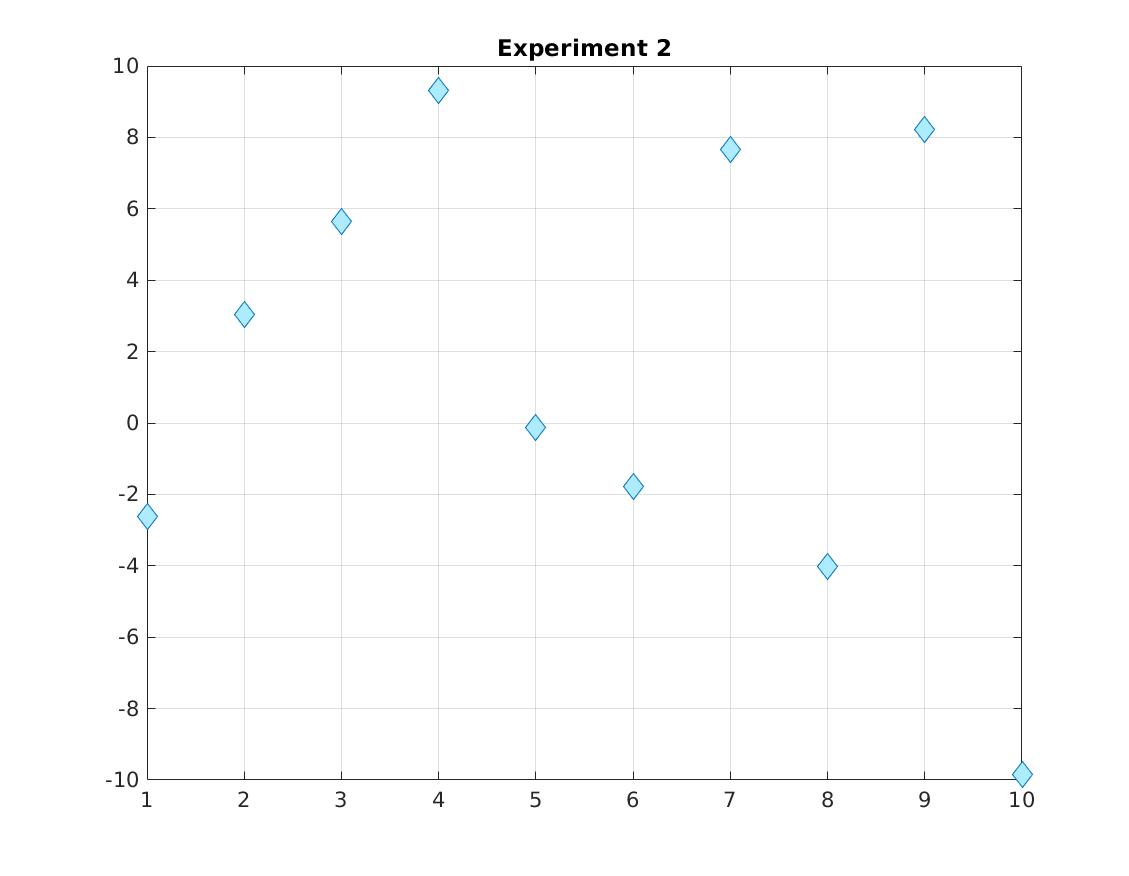}}\\
	&	{\includegraphics[width=15em]{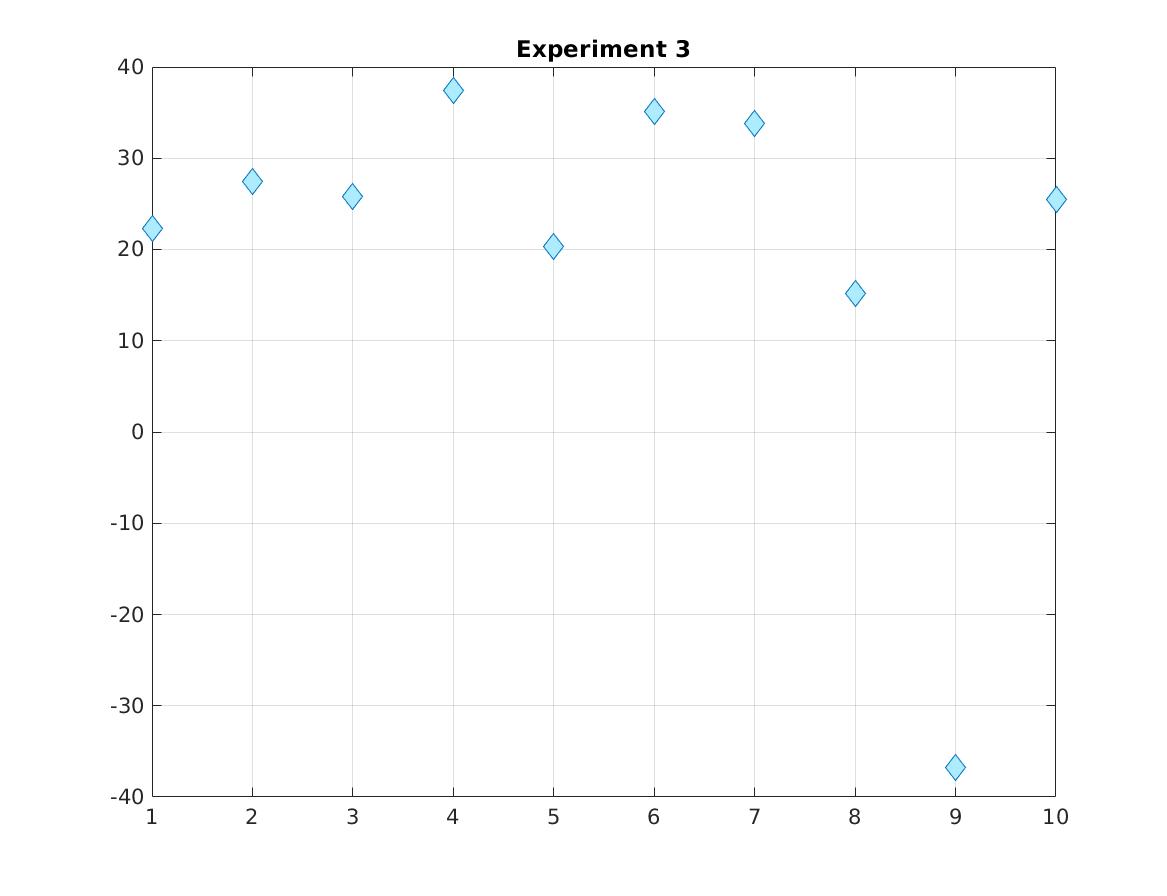}}\; 	{\includegraphics[width=15em]{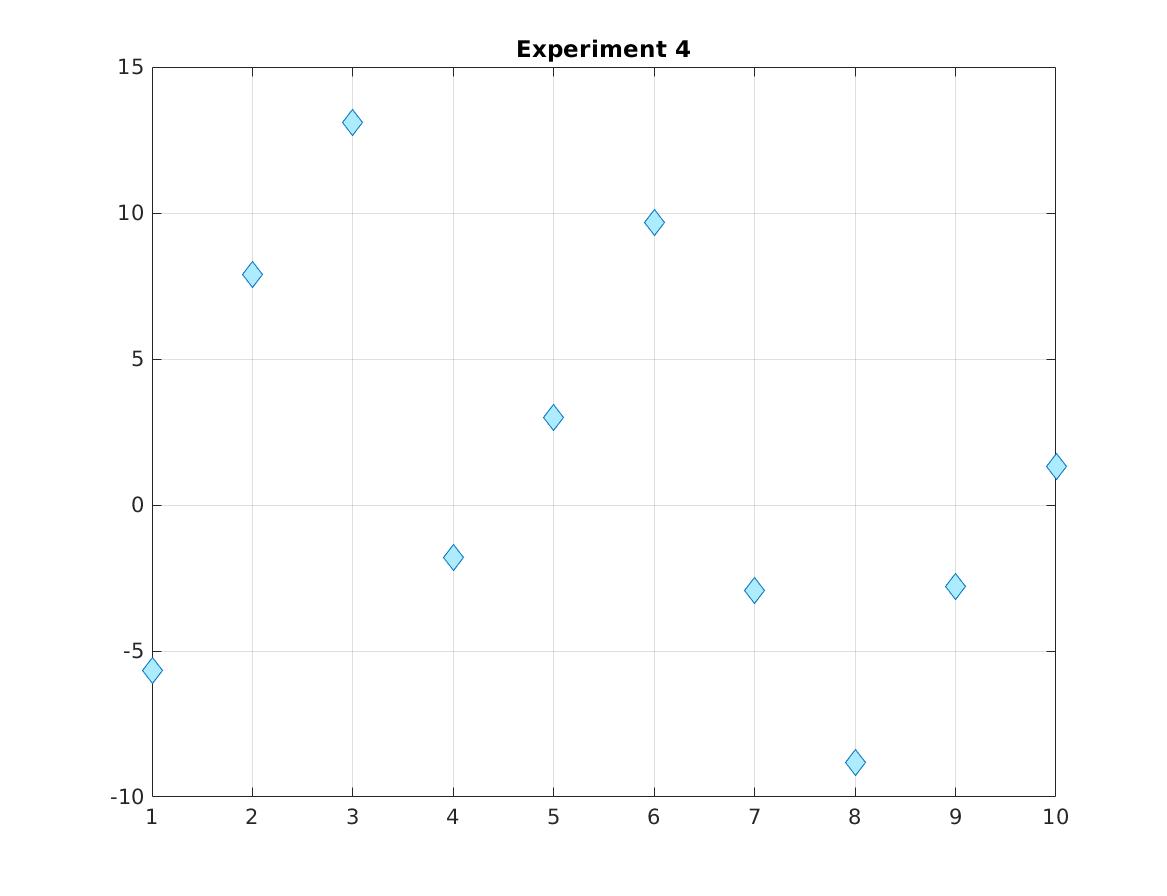}}
		\end{align*}
\caption{\footnotesize
AWS validation according to the set of thresholds is listed in Table~\ref{tab:4}. The four costs’ set which were utilized to evaluate the average value attained for each pair of threshold were: $\{r,f,\beta,\phi,h,\kappa\}=\{100,60,100,0,0,122\}$, \{100,300,100,0,0,122\}$, \{100,300,10,0,0,10\}$ and $ \{100,300,10,0,0,70\}$ for set 1 to 4, respectively.}
		\label{fig34}
	\end{center}
\end{figure}

Unsurprisingly, Figure~\ref{fig34} clearly depicts that the average value attained for each pair of thresholds is highly dependent on the set of costs. Accordingly, there is no unique pair of thresholds which is preferable to all sets. Such observation strengthens the claim that determining the thresholds \emph{a priori} compromises on performance, and implies that an adaptive mechanism, which is fed with the costs set, learns the system parameters (e.g., expected arrival rate, expected service time, expected VM deployment time, etc.) and sets the thresholds accordingly, is superior. Note that choosing inappropriate thresholds can not only result in non-optimal attainable value but can result in negative revenue (average system loss) rather than average income. For example, note that threshold set number $4$ which is the best out of all the inspected threshold pairs, for cost sets $2$ and $3$ and acceptable for cost set number $1$, resulted in revenue loss for cost set number $4$, which can be attributed to the higher percentage of rejected tasks.
Obviously, when the holding cost is high and the delay cost factor ($h$) is low, it is preferable to keep as little active VMs as possible, and deploy them on an on-demand basis (e.g., both sets $1$ and $2$ with threshold pair $8$, namely low termination threshold and high deployment threshold, keep VMs active even when the load is low which results in revenue loss due to the unnecessary holding charges). When VM deployment cost is low, deploying and terminating of VMs on a per task basis can be cost effective, depending on the arrival rate, i.e., a policy which deploys a VM upon arrival of a task and terminates the VM upon its completion can be rewarding. For instance, cost set 3 with threshold pair 8 and low holding and deployment costs.
Even though, as previously explained, there is logic for the attained revenues for each set of threshold, it can easily be seen that it is hard to predict \emph{a priori} what will be the attained value for a set of costs and a pair of thresholds. All the more so, it is hard to predict the optimal thresholds.
In order to compare the results to the optimal thresholds for each set of costs, we utilized the procedure described in Section~\ref{sec:mdp}, feeding the Bellman equation described in~\eqref{eq:b1}, all the parameters captured by the AWS (similar to the Matlab results attained in the previous section with inputs which were generated by the AWS). The results are presented in Table~\ref{tab:Costs1}. % \FMO{???}\mmme{this par. has too many fixmies and should be fixed :) }

%***\FMO{update the MATLAB results}***

\begin{table}[t]
	\vspace{0.1in}
	\center
	\begin{footnotesize}
	{\footnotesize
\begin{tabular}{|l||*{5}{c|}}\hline
&Termination threshold& Deployment threshold& Average Value \\\hline\hline
Set I & 49.8&70.1 & 18.7 \\\hline
Set II &30 &56.5 &  12.4\\\hline
Set III &18.2 &30 & 49.12\\\hline
Set IV &31.2 & 52.1& 27.12\\\hline
\end{tabular}
\normalsize}
\end{footnotesize}
	\vspace{0.1in}
	\caption{Summary of tested experiments along with the numerically attained average reward per task} \label{tab:Costs1}
\end{table}

As expected, the revenue attained by the MDP based procedure suggested herein is superior to the one attained by the best pre-defined thresholds on AWS. Specifically, we can see revenue gains ranging between $125\%$ and $200\%$ ($175\% , 125\% , 125\%$ and $200\%$ for sets $1$ to $4$, respectively). %Furthermore, even though the optimal thresholds attained by the numerically derived policy which leverages the Bellman equation were not tested on the AWS platform, we see some resemblance between them and the best ones experimented on AWS for each costs set separately.\mmme{Can we omit this sentence?} 
It can be seen that even though the pair of thresholds examined on AWS which is closest to the optimal derived ones is not necessarily the thresholds pair that provided the best revenue, the trend is similar, i.e., both deployment and termination thresholds are in the same vicinity and the gap between the two thresholds is comparable.

Although, as previously shown, the threshold based autoscale mechanism implemented on AWS maintained the trends of the optimal solution attained by the MDP solution, there are some disparities between the performance of the two. These differences cannot be bridged by trying to mimic the behavior of the optimal solution on the AWS platform, as they are associated with differences which are inherent in the mechanisms adopted by the two approaches. Note that such differences confine the AWS autoscale policy to suboptimallity. The two major essential differences are: (i) AWS Auto Scale is myopic as far as the maximal number of VMs is concerned and always performs scaling according to predefined thresholds. (ii) AWS Auto Scale uses the same thresholds regardless of the number of active VMs. Additional differences between the two schemes are given in Table~\ref{tab:3}.
\begin{table}[t]
	\vspace{0.1in}
	\center
	\begin{footnotesize}
		{\footnotesize
			\begin{tabular}{l|l|l}
				
				\hline
				$\nicefrac{method}{property}$ &	MDP Solution & AutoScaler of AWS\\ %  & comment \\
				
				\hline
				Detailed state policy &	Yes & Complex  \\
				%	&&custom implementation. \\
				%	\hline
				Awareness of VM's limit& Yes & No\\
				Awareness of deployment time &	Yes & No \\
				Load Balancing policy &	Optimal & RR only \\
				Scalability by VMs number&	BE solution is needed & Trivial \\
				Scalability by VM capacity&	BE solution is needed & Trivial \\
				Threshold alarm response & Immediate & Depends on user's budget \\
				%	&& \\
				\hline
			\end{tabular}
			\normalsize}
	\end{footnotesize}
	\vspace{0.1in}
	\caption{\scriptsize{Comparison of capabilities - policy driven from MDP and implementation on AWS. } \protect\footnotemark} \vspace{-0.3in}  \label{tab:3}
\end{table}
\footnotetext{ ELB can be immediately configured for new (not necessarily optimal) thresholds. On the other hand, our implementation for BE solution includes very careful complexity  design and scales well as long as the memory resources used to accommodate the state-space allow that. The latter, considering modern processing equipment, is never a bottleneck.}
 In order to illustrate these intrinsic differences, consider the following example where the policy for VM deployment is according to the average CPU load. Compare the two states: $q^a=\{-2,-2,4,4,4\}$ and $q^b=\{-2,-2,-2,-2,3\}$, where the limit number of tasks is $5$. The average occupancy of deployed VM in $q^a$ is $4$, while in $q^b$ it is equal to $3$. Accordingly, in the case where the deployment threshold is equal to $70\%$, there will be a deployment decision in $q^a$ and there will be no deployment decision in $q^b$. However, taking into account the non-zero deployment time, the immediate available space in $q^a$ is $3$ tasks while in $q^b$ it is equal to only $2$ tasks. Therefore, the probability of rejection is higher in state $q^b$. Hence, this policy implies a potential reward loss. Note that trying to devise a policy on AWS which takes decisions according to the detailed state is not sufficient for reaching optimality, as it requires the support of the ELB. Specifically, if the ELB scheduling method supports only Round Robin, the potential gains from supporting detailed state based policy, is quite limited.

\begin{remark}[Bellman equation for simplified state-space]
	The optimal policy found for the set of VMs represented by queues is designated to be applied to a corresponding queuing system. That is, the decisions are done according to the detailed states of all VMs rather than according to the average load at all VMs. A simple translation of each state to the load can be done, in order to apply the policy to AWS domain. This is a robust simplification, which provides an approximated load thresholds, yet serves as a good indication.
	Nevertheless, we also performed a research with specially tailored  Bellman equation, which were based on the simplified state space, such that adjustment to the AWS domain will be precise. We provide the details of the related methodology in Appendix~\ref{sec:s_BE}.
\end{remark}
%%%%%%%%%%%%%%%%%%%%%%%%%%%%%%%%%%%%%%%%%%%%%%%%%%%%%%%%%%5

%\end{document} 

\section{Related Work}\label{sec:related}  
In this section we bring a list of related works which deal with scaling and balancing in cloud computing. We also emphasize works addressing specific application realms, esp. NFV and scientific workflows.
Global point of view about scaling challenges in cloud computing was addressed in~\cite{buyya2010intercloud}. See also a review of scaling methods in~\cite{lorido2014review} and the references therein. 
A thorough discussion which accounts for many computational and cost aspects and brings specific examples of possible missions that can be offloaded to the cloud can be found in~\cite{schatz2011mpi} and in~\cite{furht2010handbook}. 
Dynamic scaling of web applications based on the number of login users to the application in a virtualized Cloud Computing environment is presented in~\cite{chieu2009dynamic}.
The recent review  on dynamic load balancing in the cloud (e.g.,~\cite{ghomi2017load}) provides a detailed coverage and splits the exiting techniques into categories. In particular, this work specifies,  general, application-oriented, network-aware category and workflow specific categories. Hence, the scope we address in this paper is highly relevant.
Challenges of VNF scheduling by means of software-defined networks (SDN) were introduced in~\cite{riera2014virtual}. 
Technical details of how the VNF could be deployed are discussed in~\cite{scholler2013resilient}.
Integration of NFV and SDN was also addressed in~\cite{costa2015sdn} in the context of 5G mobile core networks. In~\cite{yoshida2014morsa}, joint optimization of NFVI resources under
multiple constraints results in approximate algorithm which aims at improving the
capability of NFV management and orchestration. This work provides static placement of NFV-related resource.
Near optimal NP-hard placement problem is approximately solved by linear algorithm in~\cite{cohen2015near}. 
The problem of VNF scheduling aimed to optimize performance latency, was treated by linear programming in~\cite{qu2016network}.

Offloading scientific workflows to  a cloud have recently received intensive attention, see, e.g.,~\cite{zhou2016taxonomy},\cite{hoffa2008use},\cite{lin2011scheduling}. Specific examples include but not limited to cloud-based Neural Networks~\cite{teerapittayanon2017distributed}, processing of astronomy data~\cite{vockler2011experiences}, e-business workflows~\cite{xu2017near}. 

Some of the following works in queuing control can be seen as potentially applied to cloud computing. In~\cite{mace20162dfq}, queue scheduling algorithm named 2DFQ
separates requests with different
costs and different size across different worker threads. It also extends to the case where the costs are variable or unpredictable.
Optimal scheduling in hybrid cloud environment was studied in~\cite{shifrin2013optimal}. The solution provided by the authors employs MDP as well.
A great deal of work provides joint results on queue scheduling and optimization. 
We mention works which deal with fixed number of queues or queue networks, e.g.~\cite{van2000price,neely2008fairness,feitelson1997theory} and works which take this number to the limit, e.g.~\cite{ganti2007optimal,harrison2004dynamic,puhalskii2000multiclass}. See also and references therein. This work, in contrary, treat the specific scenario, where the number of queues is finite but flexible, incorporating features of deployment cost and time and termination. %These features correspond to the VNF scheduling which we aim to optimize in the sense of overall cost minimization. 
Therefore, we developed a model which captures all these features and treat it by MDP.

%In order to validate policies found from MDP, we present an AWS-based simulation set-up. 
%which includes traffic generator of computational tasks which are handled over to the finite yet controllable number of VMs by elastic load balancer (ELB) enhanced by auto-scaling option.

%While VNF tasks are suitable candidates for the presented method, the latter can be equally utilized for any computational mission. 
\section{Conclusion}
In this work, we presented a scaling and load balancing mechanism by means of optimal stochastic control, in a complex setting where cloud users (applications) aim to optimize their overall costs. The overall cost depends on several cost parameters which can be classified into two basic cost types. The first one we categorize as service revenues which are associated with SLA, namely rewards for successfully admitted tasks, fines for rejections and reward reduction for low performance. The second one we categorize as provisional costs; these refer to the Operator’s expenditures towards the cloud provider; they include cost of VM deployment and termination and cost of having VMs on-line (keep-alive cost). The controlling agent which manages both the scaling procedure and the load balancing mechanism, receive as inputs various Key Performance Indicators (KPIs) such as the load on each deployed VM, and based on the anticipated loads, the expected performance and the set of costs determine whether to deploy or terminate a VM, and to which VM to allocate an arrival task. For example, based on the system state, expected load and related performance and the anticipated value (cost) the agent can decide to stop diverting tasks to a specific VM in order to empty it and terminate its operation without dropping any tasks on the way. We formulated the problem by Markov Decision Process and derived optimal policy which optimized the application’s total cost. The value function, the optimized cost and the optimal policy (which is a direct product of it), were thoroughly investigated numerically. In particular, we numerically explored the effect of various cost parameters (e.g., holding cost and task rejection penalty), operational parameters (e.g., deployment and termination times), performance parameters (e.g., VM’s load effect on task latency) and their interdependence on the optimal policy. Via this comprehensive study, we derived several important properties of the value functions and the corresponding policies. We further examined threshold-type policies which are the common policy in operational deployments. Our work was validated through implementation of several aspects on AWS, and extracting the relevant KPIs while operating this Amazon cloud computing platform under various setups. Our results endorse, that even under this limited platform, the policy found by operating the corresponding MDP procedure suggested in this paper outperforms a-priori determined thresholds and can be utilized to determine a dynamic policy which reacts according to the actual system state. 
%\input{Conclusion}
%\iflong
\vspace*{-10pt}
\appendix
\subsection{Derivation of Bellman equation}\label{app:proof}

For succinctness, we prove for the case where $\zeta_i=0,\;\forall i$.
Note that we assume the deployments only occur at arrivals, while terminations occur only at service completions.
Write the cost function as follows:
%Write $J_i$ the cost for queue $i$
\begin{footnotesize}
	{\footnotesize
\begin{align*}
& J=\int_0^\infty e^{-\gamma t}\Big[\big(\bb(t)\cdot\beta+\bu(t)\cdot r-f*(1-\sum_i^\bq\bu_i(t))\big) dA(t)+\big(\bd(t)\cdot \psi d\bD(t)\big)+\\
&\;\;\sum_i^\bq\big(h_i(t)+\kappa*(1-\bI_i^{\bi})\big)dt\Big],
\end{align*}
\normalsize}
\end{footnotesize}
where $\bD(t)$ is a vector form of the departing processes $\{D(i)\}$.
We take infinitesimal $\theta$ such that only one or no event can happen (probability of more than
one is negligible), and use the dynamic programming principle.
\begin{footnotesize}
	{\footnotesize
\begin{align*}
&J(q(0))=\int_0^\theta e^{-\gamma t}\Big[\big(\bb(t)\cdot\beta+
\bu(t)\cdot r-f*(1-\sum_i^\bq\bu_i(t))\big) dA(t)+\big(\bd(t)\cdot\psi d\bD(t)\big)+\\
&\;\;\sum_i^\bq\big(h_i(t)+\kappa*(1-\bI_i^{\bi})\big)dt\Big]+e^{-\gamma\theta}J(q(\theta))=J^\theta+e^{-\gamma\theta}J(q(\theta)),
\end{align*}
\normalsize}
\end{footnotesize}
Expanding for all options of $J (q(\theta))$ after the time $\theta$ we have:
\begin{footnotesize}
	{\footnotesize
\begin{align}
&J^\theta+e^{-\gamma\theta}\Big[\lambda\theta J(q+\bu)+\sum_i^\bq\mu J(q-e_i)+(1-\lambda\theta-\sum_i^\bq\mu)J(q)    \Big] \label{eq:1DP}
\end{align}
\normalsize}
\end{footnotesize}
We use the following derivation:
\begin{footnotesize}
	{\footnotesize
\begin{equation}
\E\int_0^\theta e^{-\gamma t}dA(t)=\la\int_0^\theta e^{-\gamma t}dt=\la\frac{1-e^{-\gamma\theta}}{\gamma}=\la\theta,
\label{eq:eqproof-E}
\end{equation}
\normalsize}
\end{footnotesize}
In what follows we use $e^{-\gamma\theta}\simeq 1-\gamma\theta$ as $\theta\rightarrow 0$.
%The application of the above to the service processes is as follows. 
Calculate first the time used for a service in the infinitesimal interval $[0,\theta]$.
\begin{footnotesize}
	{\footnotesize
\(
T_i(t)=\int^\theta_0(1-\bI^\bee(t))(1-\bI^\bi(t))dt.
\)
\normalsize}
\end{footnotesize}
Next, denote the potential service time process which counts exponentially distributed service times as $S_i(t)$. Then,
\begin{footnotesize}
	{\footnotesize
\(
D_i(t)=S_i(T_i(t)).
\)
\normalsize}
\end{footnotesize}
Hence, as long as interval $[0,\theta]$ is small, write
\begin{footnotesize}
	{\footnotesize
\(
\E\int_0^\theta e^{-\gamma t}dD_i(t)=(1-\bI^\bee)(1-\bI^\bi)\mu_i=\tilde{\mu_i}.
\)
\normalsize}
\end{footnotesize}
Write the costs incurred in $[0,\theta]$
\begin{footnotesize}
	{\footnotesize
\begin{align*}
&J^\theta=\la\theta(\bu\cdot r-\bb\cdot \beta-f(1-\sum_i^\bq\bu_i))-\tilde\mu\theta(\bd\cdot \psi)-C(q)\theta
\end{align*}
\normalsize}
\end{footnotesize}
We now turn to the second term of~\eqref{eq:1DP}, 
\begin{footnotesize}
	{\footnotesize
\begin{align*}
&(1-\gamma\theta)\Big[\lambda\theta J(q+\bu)+\sum_i^\bq\mu J(q-e_i)+(1-\lambda\theta-\sum_i^\bn\tilde\mu)J(q)\Big]
\end{align*}
\normalsize}
\end{footnotesize}
multiply all sides by $\gamma$ and mind that $\theta^2\simeq 0$.
\begin{footnotesize}
{\footnotesize

\begin{align*}
&\gamma J(q)=\gamma\Big(\la\theta(\bu\cdot r-\bb\cdot\beta-f(1-\sum_i^\bq\bu_i))-\tilde\mu\theta(\bd\cdot\psi)-C(q)\theta\Big)
+\gamma(1-\gamma\theta)\Big[\lambda\theta J(q+\bu)\\
&+\sum_i^\bq\mu\theta J(q-e_i)+(1-\lambda\theta-\sum_i^\bq\tilde\mu\theta)J(q)\Big]
=\gamma\Big(\la\theta(\bu\cdot r-\bb\cdot\beta-f(1-\sum_i^\bq\bu_i))-\tilde\mu\theta(\bd\cdot\psi)\\
&-C(q)\theta\Big)+\Big[\lambda\theta\gamma J(q+\bu)+\sum_i^\bq\tilde\mu\theta\gamma J(q-e_i)
+(\gamma-\lambda\theta\gamma-\sum_i^\bq\tilde\mu\theta\gamma)J(q)-\gamma\theta J(q)\Big]
\end{align*}
\normalsize}
\end{footnotesize}
See that $\gamma J(q)$ on both sides cancels out and denote $\delta_q=(\sum_i^\bq\tilde\mu+\lambda+\gamma)^{-1}$.
Write:
\begin{footnotesize}
	{\footnotesize
\begin{align*}
&\gamma\theta\delta_q^{-1} J(q)=\gamma\theta\big(\la(\bu\cdot r-\bb\cdot\beta-f(1-\sum_i^\bq\bu_i))-\tilde\mu(\bd\cdot\psi)-C(q)\big)+\gamma\theta\Big[\lambda J(q+\bs)+\sum_i^\bq\tilde\mu J(q-e_i) \Big]
\end{align*}
\normalsize}
\end{footnotesize}
Divide all by $\gamma\theta$ and multiply by $\delta_q$
\begin{footnotesize}
	{\footnotesize
\begin{align*}
& J(q)=\big(\la(\bu\cdot r-\bb\cdot \beta-f(1-\sum_i^\bq\bu_i))-\tilde\mu(\bd\cdot\psi)-C(q)\big)\delta_q+\Big[\lambda J(q+\bu)+\sum_i^\bq\tilde\mu J(q-e_i) \Big]\delta_q
\end{align*}
\normalsize}
\end{footnotesize}
Arrange according to the processes:
\begin{footnotesize}
	{\footnotesize
\begin{align*}
& J(q)=\delta_q\la\big[\bu\cdot r-\bb\cdot\beta-f(1-\sum_i^\bq\bu_i)+J(q+\bu)\big]+\Big[\sum_i^\bq\tilde\mu J(q-e_i)-\tilde\mu(\bd\cdot\psi) \Big]\delta_q-C(q)\delta_q
\end{align*}
\normalsize}
\end{footnotesize}
Observe that by definition, $\bu,\bb,\bd$ incorporate (in corresponding feasible cases) the scheduling, build and destroy decisions, respectively. Hence, this equation is identical to~\eqref{eq:b}, where the feasibility is ensured by the corresponding indicators.
\subsection{Bellman equations for simplified state-spaces  }\label{sec:s_BE}
We provide Bellman equations for simplified state-space, such that the optimal policy is directly derived in terms of the average load on all VMs. %We provide two possible implementation options.
%Since this part is mainly utilized as proof of concept (POC), and in order to simplify the procedure, we rely only on a single performance indicator (CPU load),
% ***\FMO{Furthermore, we omitted all the other parameters that we cannot extract from the system and utilized a simplified Bellman equation, in which the value is a function of only the average load and of the number of the active VMs: ***Did we?}***
Consider first a two dimensional state space denoted by $\calS_1$, such that $\calS_1=\calM\times\calL\times\calZ$,
where 
\[
\calM=\{1,\cdots,M_{max}\}\;, \calZ=\{1,\cdots,M_{max}\}\; \text{  and  } \calL=\{0,l,2l,\cdots,L_{max}\}
\]
$M_{max}$ denotes the maximal number of active VMs and $L$ represents deployed VM utilization, i.e., $L$ denotes the average load on all active VMs in units of percentage of the maximal load.  For instance, if there is a single active VM which can accommodate up to 4 tasks, serving a single task, then $L=\frac{1}{4}=25\%$, if there are three active VMs each potentially accommodating up to 4 tasks,  serving altogether 6 tasks (regardless of how many each one serves), the load is: $L=\frac{6}{3\times4}=50\%$.
$L_{max}$ corresponds to the maximal load of $100\%$, and $l$ is the resolution of the space $\calL$, which, for our purposes was chose to correspond to $5\%$ or $10\%$.
The spaces $\calM$ and $\calZ$ stand for the number of active and deploying (i.e., being in the process of deployment but not yet ready to accept tasks) VMs. Clearly it holds $Z+M\leq M_{max}$.
Note that the minimal number of active VMs is set to $1$, which corresponds to the AWS property of preserving at least one active VM. Also note that VMs in the process of deployment are considered as active with zero load. 

Recall that $r$ and $\beta$ denote the reward for admitted tasks and the deployment cost, respectively. 
Also recall the notations for actions $\bb$ and $\bd$, standing for deployment of a new VM and termination of a VM.

Consider $C(L,M,Z)$), the cost per unit of time associated with having $M$ active machines with average load $L$ and $Z$ deploying machines. Observe that upon task service or admission the load does not necessarily increases or decreases by $1$. 
Hence we introduce transition probabilities $p((L',M',Z')|(L,M,Z),\bb)$ and $p((L',M',Z')|(L,M,Z),\bd$). 
The straightforward method to obtain these probabilities and reward $C(L,M,Z)$) for all states is by learning through simulating of the optimal policy for the original queuing system with the optimal policy for the full state-space. Once the learning is done the value function of the simplified state-space can be calculated according to the following:
%The value  associated with state $(L,M,Z)$ is written 
\begin{footnotesize}
	{\footnotesize
		\begin{align}
		&V(L,M,Z)=\big[C(L,M,Z)+\lambda \max_{\bb}\{\sum_{L'}p((L',M,Z)|(L,M,Z),\bb=0)V(L',M,Z)+r,\label{eq:SB1}\\
		&\;\;\;\qquad\qquad\qquad\qquad\qquad\qquad\qquad\sum_{L'}p((L',M,Z+1)|(L,M,Z),\bb=1)V(L',M,Z+1)+r-\beta\}\label{eq:SB2}\\
		&+\;\;\zeta_{L,M,Z}V(L,M+1,Z-1)\label{eq:SB3}\\
		&\;\;+\mu_{L,M,Z}\max\{\sum_{L'}p((L',M,Z)|(L,M,Z),\bd=0)V(L',M,Z),\label{eq:SB5}\\
	&\;\;\;\qquad\qquad\qquad	\sum_{L'}p((L',M-1,Z)|(L,M,Z)\bd=1)V(L',M-1,Z)\}\big]\delta_{L,M,Z}\label{eq:SB4}
		\end{align}
		\normalsize}
\end{footnotesize}
Note that the rates service rates $\mu_{L,M,Z}$ and VM deployment reciprocal time $\zeta_{L,M,Z}$ are learned within the same learning process. 
In~\ref{eq:SB1} and~\ref{eq:SB2} the selection regarding opening a new VM is written. Part~\ref{eq:SB3} refers to the new VM deployment process. The selection if to terminate a VM upon service completion is written in~\ref{eq:SB5} and~\ref{eq:SB4}.
For the sake of succinctness,  we omitted here the boundary conditions of the Bellman equation above. 

In order to further adjust to AWS limitations the space $\calM$ and the space $\calZ$ can be degenerated to 
$\calM=\{0,1\}$, $\calZ=\{0,1\}$ thus being agnostic to the number of deployed VMs and considering only the information if the maximal number of VMs has been deployed or not and if there is a VM in deployment process.
In this case, the entire dynamics of the value function and the transition probabilities will be solely expressed by the load. 
%\ref{eq_Cq} 

\subsection{Proof of Lemma~\ref{lem:dom} }\label{sec:lemdom}
%\begin{lemma}[Value function domination]
We show that for any $q^a\succeq q^b$ it holds $V(q^a)\leq V(q^b)$.
%\end{lemma}
\begin{proof}
To show the statement define two systems, denoted by $a$ and $b$. At time $t$ the state of each system is denoted by $a_t$ and $b_t$. Next, we stochastically couple these systems. Namely, we apply the same policy to both states, denoted as $\pi_{ab}$. The policy is optimal for the system $a$, while system $b$ mimics all actions from system $a$. In particular the mimicking policy follows the following rules:
\begin{enumerate}
	\item A task which is scheduled in $a$ is scheduled in $b$ into the same queue.
	\item In the case where the rejection optimally applies in system $a$, the task is rejected in $b$ as well.
	\item In the case where there is a forced rejection (that is, the corresponding queue in $a$ was full): 
	\begin{itemize}
		\item In the case the same queue in $b$ was also full, the rejection applies in $b$ as well
		\item In the case that queue in $b$ was not full - the task is scheduled and the occupancy difference between these two queues is reduced by $1$.
	\end{itemize} 
\item Each served task departs concurrently from $a$ and $b$.
\item In the case where during a departure in $a$ the same queue in $b$ was empty - nothing happens in $b$. (Hence, system $b$ will keep alive queues even if empty as long as they are kept in $a$).
\item  The "terminate" and "deploy" decisions in $b$ are executed concurrently with $a$. 
\end{enumerate}  
Observe that eventually the systems become identical according to one of the following two cases:
\begin{itemize}
	\item The queues which were unequal in the systems become empty
	\item The queue which were unequal in the systems become full 
\end{itemize}
Now denote the first time $\tau$ when both states become equal, that is $a_\tau=b_\tau$. 
Clearly, by definition of $\pi_{ab}$ it holds $e^{-\gamma \tau}V(a_\tau)=e^{-\gamma \tau}V(b_\tau)$.
Next, as for the rewards in time interval $t\in[0,\tau)$, it holds $r(a_t)\leq r(b_t)$ and as for the fines it holds $f(a_t)\geq f(b_t)$. Finally, by $\pi_{ab}$, all holding costs and costs incurred by deploy and terminate actions which accumulate to the total cost %according to~\eqref{eq:3} 
are higher in $a$.
%Denote  it always holds in $[a,\tau)$ $r(a)=r(b)$. This is because starting at $a$ one gets all the rewards as in $b$, but can pay less delay and holding cost and less fines.
Therefore,
\begin{footnotesize}
	{\footnotesize
\(
V(b_0)\leq J^{\pi_{ab}}(b_0)\leq V(a_0),
\)
\normalsize}
\end{footnotesize}
Hence, the statement in the lemma follows.\qed
\end{proof}

\bibliographystyle{abbrv}
\bibliography{nfv}

\begin{thebibliography}{10}

\bibitem{scprot}
{AWS Auto-Scaling user guide }.
\newblock
  \url{https://docs.aws.amazon.com/autoscaling/ec2/userguide/as-instance-termination.html},
  2018.

\bibitem{drain}
{ELB Connection Draining}.
\newblock
  \url{https://aws.amazon.com/blogs/aws/elb-connection-draining-remove-instances-from-service-with-care/},
  2018.

\bibitem{awsAS}
{Amazon Documentation}.
\newblock {AWS autoscaling user guide}.
\newblock \url{https://docs.aws.amazon.com/autoscaling/latest/userguide}, 2017.

\bibitem{arfeen2011framework}
M.~A. Arfeen, K.~Pawlikowski, and A.~Willig.
\newblock A framework for resource allocation strategies in cloud computing
  environment.
\newblock In {\em Computer Software and Applications Conference Workshops
  (COMPSACW), 2011 IEEE 35th Annual}, pages 261--266. IEEE, 2011.

\bibitem{bertsekas1995dynamic}
D.~Bertsekas.
\newblock {\em Dynamic programming and optimal control}, volume~2.
\newblock Athena Scientific Belmont, MA, 1995.

\bibitem{buyya2010intercloud}
R.~Buyya, R.~Ranjan, and R.~N. Calheiros.
\newblock Intercloud: Utility-oriented federation of cloud computing
  environments for scaling of application services.
\newblock In {\em International Conference on Algorithms and Architectures for
  Parallel Processing}, pages 13--31. Springer, 2010.

\bibitem{chieu2009dynamic}
T.~C. Chieu, A.~Mohindra, A.~A. Karve, and A.~Segal.
\newblock Dynamic scaling of web applications in a virtualized cloud computing
  environment.
\newblock In {\em E-Business Engineering, 2009. ICEBE'09. IEEE International
  Conference on}, pages 281--286. IEEE, 2009.

\bibitem{cohen2015near}
R.~Cohen, L.~Lewin-Eytan, J.~S. Naor, and D.~Raz.
\newblock Near optimal placement of virtual network functions.
\newblock In {\em Computer Communications (INFOCOM), 2015 IEEE Conference on},
  pages 1346--1354. IEEE, 2015.

\bibitem{costa2015sdn}
Costa-Requena et~al.
\newblock Sdn and nfv integration in generalized mobile network architecture.
\newblock In {\em Networks and Communications (EuCNC), 2015 European Conference
  on}, pages 154--158. IEEE.

\bibitem{feitelson1997theory}
D.~G. Feitelson, L.~Rudolph, U.~Schwiegelshohn, K.~C. Sevcik, and P.~Wong.
\newblock Theory and practice in parallel job scheduling.
\newblock In {\em Workshop on Job Scheduling Strategies for Parallel
  Processing}, pages 1--34. Springer, 1997.

\bibitem{furht2010handbook}
B.~Furht and A.~Escalante.
\newblock {\em Handbook of cloud computing}, volume~3.
\newblock Springer, 2010.

\bibitem{ganti2007optimal}
A.~Ganti, E.~Modiano, and J.~N. Tsitsiklis.
\newblock Optimal transmission scheduling in symmetric communication models
  with intermittent connectivity.
\newblock {\em IEEE Transactions on Information Theory}, 53(3):998--1008, 2007.

\bibitem{ghomi2017load}
E.~J. Ghomi, A.~M. Rahmani, and N.~N. Qader.
\newblock Load-balancing algorithms in cloud computing: A survey.
\newblock {\em Journal of Network and Computer Applications}, 88:50--71, 2017.

\bibitem{harrison2004dynamic}
J.~M. Harrison and A.~Zeevi.
\newblock Dynamic scheduling of a multiclass queue in the halfin-whitt heavy
  traffic regime.
\newblock {\em Operations Research}, 52(2):243--257, 2004.

\bibitem{hoffa2008use}
C.~Hoffa, G.~Mehta, T.~Freeman, E.~Deelman, K.~Keahey, B.~Berriman, and
  J.~Good.
\newblock On the use of cloud computing for scientific workflows.
\newblock In {\em eScience, 2008. eScience'08. IEEE Fourth International
  Conference on}, pages 640--645. IEEE, 2008.

\bibitem{li2011cloud}
K.~Li, G.~Xu, G.~Zhao, Y.~Dong, and D.~Wang.
\newblock Cloud task scheduling based on load balancing ant colony
  optimization.
\newblock In {\em 2011 Sixth Annual ChinaGrid Conference}, pages 3--9. IEEE,
  2011.

\bibitem{lin2011scheduling}
C.~Lin and S.~Lu.
\newblock Scheduling scientific workflows elastically for cloud computing.
\newblock In {\em Cloud Computing (CLOUD), 2011 IEEE International Conference
  on}, pages 746--747. IEEE, 2011.

\bibitem{lorido2014review}
T.~Lorido-Botran, J.~Miguel-Alonso, and J.~A. Lozano.
\newblock A review of auto-scaling techniques for elastic applications in cloud
  environments.
\newblock {\em Journal of Grid Computing}, 12(4):559--592, 2014.

\bibitem{mace20162dfq}
J.~Mace, P.~Bodik, M.~Musuvathi, R.~Fonseca, and K.~Varadarajan.
\newblock 2dfq: Two-dimensional fair queuing for multi-tenant cloud services.
\newblock In {\em Proceedings of the 2016 conference on ACM SIGCOMM 2016
  Conference}, pages 144--159.

\bibitem{neely2008fairness}
M.~J. Neely, E.~Modiano, and C.-P. Li.
\newblock Fairness and optimal stochastic control for heterogeneous networks.
\newblock {\em IEEE/ACM Transactions On Networking}, 16(2):396--409.

\bibitem{puhalskii2000multiclass}
A.~A. Puhalskii, M.~I. Reiman, et~al.
\newblock The multiclass gi/ph/n queue in the halfin-whitt regime.
\newblock {\em Advances in Applied Probability}, 32(2):564--595, 2000.

\bibitem{qu2016network}
L.~Qu, C.~Assi, and K.~Shaban.
\newblock Network function virtualization scheduling with transmission delay
  optimization.
\newblock In {\em Network Operations and Management Symposium (NOMS), 2016
  IEEE/IFIP}, pages 638--644.

\bibitem{riera2014virtual}
J.~F. Riera, E.~Escalona, J.~Batall{\'e}, E.~Grasa, and J.~A.
  Garc{\'\i}a-Esp{\'\i}n.
\newblock Virtual network function scheduling: Concept and challenges.
\newblock In {\em Smart Communications in Network Technologies (SaCoNeT), 2014
  International Conference on}, pages 1--5. IEEE.

\bibitem{schatz2011mpi}
F.~Schatz, S.~Koschnicke, N.~Paulsen, C.~Starke, and M.~Schimmler.
\newblock Mpi performance analysis of amazon ec2 cloud services for high
  performance computing.
\newblock In {\em Advances in Computing and Communications}, pages 371--381.
  Springer, 2011.

\bibitem{scholler2013resilient}
M.~Sch{\"o}ller, M.~Stiemerling, A.~Ripke, and R.~Bless.
\newblock Resilient deployment of virtual network functions.
\newblock In {\em 2013 5th International Congress on Ultra Modern
  Telecommunications and Control Systems and Workshops (ICUMT)}, pages
  208--214. IEEE.

\bibitem{shifrin2013optimal}
M.~Shifrin, R.~Atar, and I.~Cidon.
\newblock Optimal scheduling in the hybrid-cloud.
\newblock In {\em 2013 IFIP/IEEE International Symposium on Integrated Network
  Management (IM 2013)}, pages 51--59.

\bibitem{teerapittayanon2017distributed}
S.~Teerapittayanon, B.~McDanel, and H.~Kung.
\newblock Distributed deep neural networks over the cloud, the edge and end
  devices.
\newblock In {\em Distributed Computing Systems (ICDCS), 2017 IEEE 37th
  International Conference on}, pages 328--339. IEEE, 2017.

\bibitem{tesauro2006hybrid}
G.~Tesauro, N.~K. Jong, R.~Das, and M.~N. Bennani.
\newblock A hybrid reinforcement learning approach to autonomic resource
  allocation.
\newblock In {\em Autonomic Computing, 2006. ICAC'06. IEEE International
  Conference on}, pages 65--73. IEEE, 2006.

\bibitem{van2000price}
J.~A. Van~Mieghem.
\newblock Price and service discrimination in queuing systems: Incentive
  compatibility of gc $\mu$ scheduling.
\newblock {\em Management Science}, 46(9):1249--1267, 2000.

\bibitem{vockler2011experiences}
J.-S. V{\"o}ckler, G.~Juve, E.~Deelman, M.~Rynge, and B.~Berriman.
\newblock Experiences using cloud computing for a scientific workflow
  application.
\newblock In {\em Proceedings of the 2nd international workshop on Scientific
  cloud computing}, pages 15--24. ACM, 2011.

\bibitem{xu2017survey}
M.~Xu, W.~Tian, and R.~Buyya.
\newblock A survey on load balancing algorithms for virtual machines placement
  in cloud computing.
\newblock {\em Concurrency and Computation: Practice and Experience}, 29(12),
  2017.

\bibitem{xu2017near}
R.~Xu, Y.~Wang, W.~Huang, D.~Yuan, Y.~Xie, and Y.~Yang.
\newblock Near-optimal dynamic priority scheduling strategy for
  instance-intensive business workflows in cloud computing.
\newblock {\em Concurrency and Computation: Practice and Experience}, 2017.

\bibitem{yoshida2014morsa}
M.~Yoshida, W.~Shen, T.~Kawabata, K.~Minato, and W.~Imajuku.
\newblock Morsa: A multi-objective resource scheduling algorithm for nfv
  infrastructure.
\newblock In {\em Network Operations and Management Symposium (APNOMS), 2014
  16th Asia-Pacific}, pages 1--6. IEEE.

\bibitem{zhou2016taxonomy}
A.~C. Zhou, B.~He, and S.~Ibrahim.
\newblock A taxonomy and survey of scientific computing in the cloud, 2016.

\end{thebibliography}

%\appendix

%\bibliography{ncmdp-I}
\end{document}